\documentclass[11pt,a4paper]{article}
\pdfoutput=1

\usepackage[utf8]{inputenc}

\usepackage{mathpazo}
\usepackage{geometry,calc}
\usepackage[paper=a4,DIV=11]{typearea}

\usepackage{calc}

\usepackage{wrapfig}
\usepackage[font=small,labelfont=it]{caption}

\usepackage[bottom]{footmisc}

\usepackage{fancyhdr,lastpage}
\usepackage[nodayofweek]{datetime}
\usepackage[super]{nth}
\usepackage{paralist}
\usepackage{comment}

\usepackage{xcolor}
\usepackage{doi}
\usepackage{xstring}
\makeatletter
\renewcommand{\doi}[1]{%
  \begingroup
    \let\bibinfo\@secondoftwo%
    \urlstyle{tt}%
    \href{http://dx.doi.org/#1}{%
      \textsc{doi:}
      \nolinkurl{#1}%
    }%
  \endgroup%
}

\providecommand{\acmid}[1]{%
  \urlstyle{tt}%
  \href{http://dl.acm.org/citation.cfm?id=#1}{%
    \textsc{acmid:}
    \nolinkurl{#1}}}
\providecommand{\arxiv}[1]{%
  \urlstyle{tt}%
  \href{http://arxiv.org/abs/#1}{%
    \textsc{arXiv:}
    \nolinkurl{#1}}}
\providecommand{\halinria}[1]{%
  \urlstyle{tt}%
  \href{http://hal.inria.fr/hal-#1}{%
    \textsc{HAL INRIA:}
    \nolinkurl{#1}}}
\providecommand{\halarchivesouvertes}[1]{%
  \urlstyle{tt}%
  \href{http://hal.archives-ouvertes.fr/hal-#1}{%
    \textsc{HAL Archives Ouvertes:}
    \nolinkurl{#1}}}
\providecommand{\urn}[1]{%
  \urlstyle{tt}%
  \href{http://nbn-resolving.de/#1}{%
    \textsc{URN:}
    \nolinkurl{#1}}}

\providecommand{\eprint}[2]{%
  \begingroup%
    \let\bibinfo\@secondoftwo%
    \expandafter\csname#1\endcsname{#2}%
  \endgroup%
}
\makeatother

\usepackage[%
hyperref=true,
]{biblatex}
\DeclareFieldFormat{eprint:arxiv}{%
  \textsc{arXiv\addcolon\space}%
  \ifhyperref
  {\href{http://arxiv.org/abs/#1}{\nolinkurl{#1}}}
  {\nolinkurl{#1}}}
\DeclareFieldFormat{eprint:citeseerx}{%
  \textsc{CiteSeerX\addcolon\space}%
  \ifhyperref
  {\href{http://citeseerx.ist.psu.edu/viewdoc/summary?doi=#1}{\nolinkurl{#1}}}
  {\nolinkurl{#1}}}
\DeclareFieldFormat{eprint:hal}{%
  \textsc{oai\addcolon\space}%
  \ifhyperref
  {\href{http://hal.inria.fr/#1}{\nolinkurl{hal.inria.fr:#1}}}
  {\nolinkurl{hal.inria.fr:#1}}}
\DeclareFieldFormat{eprint:acmid}{%
  \textsc{acmid\addcolon\space}%
  \ifhyperref
  {\href{http://dl.acm.org/citation.cfm?id=#1}{\nolinkurl{#1}}}
  {\nolinkurl{#1}}}
\DeclareFieldFormat{eprint:urn}{%
  \textsc{urn\addcolon\space}%
  \ifhyperref
  {\href{http://nbn-resolving.de/#1}{\nolinkurl{#1}}}
  {\nolinkurl{#1}}}

\DeclareBibliographyDriver{inproceedings}{%
  \usebibmacro{bibindex}%
  \usebibmacro{begentry}%
  \usebibmacro{author/translator+others}%
  \setunit{\labelnamepunct}\newblock
  \usebibmacro{title}%
  \newunit
  \printlist{language}%
  \newunit\newblock
  \usebibmacro{byauthor}%
  \newunit\newblock
  \usebibmacro{in:}%
  \iffieldundef{crossref}
  {\usebibmacro{crossref:full}}
  {\usebibmacro{crossref:label}}
  \newunit\newblock
  \usebibmacro{chapter+pages}%
  \iffieldundef{crossref}
  {\usebibmacro{crossref:extrainfo}}
  {}
  \setunit{\bibpagerefpunct}\newblock
  \usebibmacro{pageref}%
  \usebibmacro{finentry}}

\newbibmacro{crossref:full}{%
  \usebibmacro{maintitle+booktitle}%
  \newunit\newblock
  \usebibmacro{event+venue+date}%
  \newunit\newblock
  \usebibmacro{byeditor+others}%
  \newunit\newblock
  \iffieldundef{maintitle}
  {\printfield{volume}%
    \printfield{part}}
  {}%
  \newunit
  \printfield{volumes}%
  \newunit\newblock
  \usebibmacro{series+number}%
  \newunit\newblock
  \printfield{note}%
  \newunit\newblock
  \printlist{organization}%
  \newunit
  \usebibmacro{publisher+location+date}}

\newbibmacro{crossref:label}{%
  \entrydata{\strfield{crossref}}
  {\printtext{\mkbibbrackets
      {\printfield{labelnumber}}}}}

\newbibmacro{crossref:extrainfo}{%
  \newunit\newblock
  \iftoggle{bbx:isbn}
  {\printfield{isbn}}
  {}%
  \newunit\newblock
  \usebibmacro{doi+eprint+url}%
  \newunit\newblock
  \usebibmacro{addendum+pubstate}}%

\bibliography{CBiSolve_complexity}

\usepackage{color}
\definecolor{citecolor}{rgb}{0.8,0,0}
\definecolor{linkcolor}{rgb}{0,0,0.8}
\definecolor{urlcolor}{rgb}{0,0,0.8}

\usepackage[fleqn]{amsmath}
\usepackage{amsthm,amssymb}
\usepackage{bold-extra}

\usepackage[capitalise]{cleveref}
\hypersetup{%
  colorlinks=true,
  linktocpage=true,
  citecolor=citecolor,
  linkcolor=linkcolor,
  urlcolor=urlcolor,
  bookmarksopen=true,
  unicode
}
\usepackage{relsize}

\usepackage{mdwlist}

\usepackage[
switch,
mathlines,
]{lineno}
\usepackage{color}

\makeatletter
\newcommand\linenomathWithnumbersforams{%
  \ifLineNumbers
  \ifnum\interlinepenalty>-\linenopenaltypar
  \global\holdinginserts\thr@@
  \advance\interlinepenalty \linenopenalty
  \ifhmode                                   
  \advance\predisplaypenalty \linenopenalty
  \fi
  \advance\interdisplaylinepenalty \linenopenalty
  \fi
  \fi
  \ignorespaces
}

\newcommand\linenomathWithnumbersformultline{%
  \ifLineNumbers
  \ifnum\interlinepenalty>-\linenopenaltypar
  \global\holdinginserts\thr@@
  \advance\interlinepenalty \linenopenalty
  \ifhmode                                   
  \fi
  \advance\interdisplaylinepenalty \linenopenalty
  \fi
  \fi
  \ignorespaces
}

\newcommand*\patchAmsMathEnvironmentForLineno[1]{%
  \expandafter\let\csname old#1\expandafter\endcsname\csname #1\endcsname
  \expandafter\let\csname oldend#1\expandafter\endcsname\csname end#1\endcsname
  \renewenvironment{#1}%
  {\def\linenomath{\linenomathWithnumbersforams}
    \@namedef{linenomath*}{\linenomathNonumbers}
    \linenomath\csname old#1\endcsname}%
  {\csname oldend#1\endcsname\endlinenomath}}%
\newcommand*\patchBothAmsMathEnvironmentsForLineno[1]{%
  \patchAmsMathEnvironmentForLineno{#1}%
  \patchAmsMathEnvironmentForLineno{#1*}}%

\newcommand*\origPatchAmsMathEnvironmentForLineno[1]{%
  \expandafter\let\csname old#1\expandafter\endcsname\csname #1\endcsname
  \expandafter\let\csname oldend#1\expandafter\endcsname\csname end#1\endcsname
  \renewenvironment{#1}%
  {\linenomath\csname old#1\endcsname}%
  {\csname oldend#1\endcsname\endlinenomath}}%
\newcommand*\origPatchBothAmsMathEnvironmentsForLineno[1]{%
  \origPatchAmsMathEnvironmentForLineno{#1}%
  \origPatchAmsMathEnvironmentForLineno{#1*}}%

\makeatother

\AtBeginDocument{%
  \origPatchBothAmsMathEnvironmentsForLineno{equation}%
  \patchBothAmsMathEnvironmentsForLineno{align}%
  \patchBothAmsMathEnvironmentsForLineno{flalign}%
  \patchBothAmsMathEnvironmentsForLineno{alignat}%
  \patchBothAmsMathEnvironmentsForLineno{gather}%
  \patchBothAmsMathEnvironmentsForLineno{multline}%
  \patchAmsMathEnvironmentForLineno{split}%
}

\usepackage{microtype}
\usepackage[mathic]{mathtools}

\makeatletter
\def\th@plain{%
  \thm@notefont{}
  \itshape 
}
\def\th@definition{%
  \thm@notefont{}
  \normalfont 
}
\makeatother

\usepackage{thmtools}
\usepackage{thm-restate}
\newtheorem{theorem}{Theorem}
\Crefname{theorem}{Theorem}{Theorems}

\newtheorem{corollary}[theorem]{Corollary}
\Crefname{corollary}{Corollary}{Corollaries}

\newtheorem{lemma}[theorem]{Lemma}
\Crefname{lemma}{Lemma}{Lemmata}

\theoremstyle{definition}
\newtheorem{definition}[theorem]{Definition}
\Crefname{definition}{Definition}{Definitions}

\theoremstyle{remark}

\Crefname{remark}{Remark}{Remarks}

\theoremstyle{definition}
\newtheorem{algorithm}[theorem]{Algorithm}
\Crefname{algorithm}{Algorithm}{Algorithms}

\Crefname{proof}{Proof}{Proofs}

\Crefname{figure}{Figure}{Figures}

\crefformat{equation}{#2(#1)#3}
\Crefformat{equation}{Equation #2(#1)#3}
\crefmultiformat{equation}{#2(#1)#3}{ and~#2(#1)#3}{, #2(#1)#3}{ and~#2(#1)#3}
\Crefmultiformat{equation}{Equations #2(#1)#3}{ and~#2(#1)#3}{, #2(#1)#3}{ and~#2(#1)#3}
\crefrangeformat{equation}{#3(#1)#4--#5(#2)#6}
\Crefrangeformat{equation}{Equations #3(#1)#4 to #5(#2)#6}

\crefformat{enumi}{#2(#1)#3}
\crefformat{enumii}{#2(#1)#3}
\crefformat{enumiii}{#2(#1)#3}
\crefformat{enumiv}{#2(#1)#3}

\makeatletter
\def\@itemref#1#2:#3\relax{#1{#2}\;\cref{#2:#3}}
\newcommand{\itemref}[1]{\@itemref\cref#1\relax}
\newcommand{\Itemref}[1]{\@itemref\Cref#1\relax}
\makeatother

\Crefname{section}{Section}{Sections}
\crefformat{appendix}{#2#1#3}

\newtheorem*{keywords}{Keywords}

\makeatletter
\newcommand*{\IfSmallCapsTF}{%
  \ifx\f@shape\my@test@sc
    \expandafter\@firstoftwo
  \else
    \expandafter\@secondoftwo
  \fi
}
\newcommand*{\my@test@sc}{sc}
\newcommand*{\IfBoldTF}{%
  \ifx\f@series\my@test@bold
    \expandafter\@firstoftwo
  \else
    \expandafter\@secondoftwo
  \fi
}
\newcommand*{\my@test@bold}{bx}
\makeatother

\newcommand{\NN}{\ensuremath{\mathbb{N}}}
\newcommand{\ZZ}{\ensuremath{\mathbb{Z}}}

\newcommand{\RR}{\ensuremath{\mathbb{R}}}
\newcommand{\CC}{\ensuremath{\mathbb{C}}}
\newcommand{\DD}{\ensuremath{\mathbb{D}}}
\newcommand{\FF}{\ensuremath{\mathbb{F}}}
\newcommand{\Oh}{\ensuremath{O}}
\newcommand{\sOh}{\ensuremath{{\tilde{\Oh}}}}
\newcommand{\ii}{\ensuremath{\mathrm{i}}}

\DeclareMathOperator{\LOG}{\ensuremath{\operatorname{log\mathllap{\raisebox{1.38ex}{\rule{1.03em}{0.08ex}}\hspace{0.025em}}}}}
\DeclareMathOperator{\crampedLOG}{\ensuremath{\operatorname{log\mathllap{\raisebox{1.03ex}{\rule{0.75em}{0.06ex}}\hspace{0.022em}}}}}

\DeclareMathOperator{\had}{\mathcal{H}}

\DeclareMathOperator{\lb}{lb}
\DeclareMathOperator{\LB}{LB}
\DeclareMathOperator{\ub}{ub}
\DeclareMathOperator{\UB}{UB}

\DeclareMathOperator{\lcf}{lcf}
\DeclareMathOperator{\mult}{mult}
\DeclareMathOperator{\Mea}{\mathcal{M}}
\DeclareMathOperator{\res}{res}
\DeclareMathOperator{\separ}{sep}
\DeclareMathOperator{\rev}{rev}
\DeclareMathOperator{\mul}{Mul}

\DeclarePairedDelimiter\set{\{}{\}}
\DeclarePairedDelimiter\abs{\lvert}{\rvert}
\DeclarePairedDelimiter\norm{\lVert}{\rVert}
\DeclarePairedDelimiter\sumnorm{\lVert}{\rVert_1}
\DeclarePairedDelimiter\infnorm{\lVert}{\rVert_\infty}
\DeclarePairedDelimiter\ceil{\lceil}{\rceil}

\DeclareMathOperator{\bdiv}{div}
\newcommand{\tf}{\ensuremath{\tilde{f}}}

\newcommand{\tg}{\ensuremath{\tilde{g}}}

\newcommand{\tH}{\ensuremath{\tilde{H}}}

\newcommand{\tr}{\ensuremath{\tilde{r}}}

\usepackage{xspace}
\newcommand{\bisolve}{\textsc{BiSolve}\xspace}
\newcommand{\cbisolve}{\texorpdfstring{\(\CC\)}{ℂ}\textsc{BiSolve}\xspace}

\let\originalparagraph\paragraph
\renewcommand{\paragraph}[1]{\originalparagraph*{#1}\phantomsection}
\newcommand{\newparagraph}{\pagebreak[1]\vspace\baselineskip\par\noindent}
\newenvironment{inlineableEquation}{\csname align*\endcsname}{\csname endalign*\endcsname}
\renewenvironment{inlineableEquation}{\(}{\)}

\usepackage{enumitem}
\newlist{well-isol}{enumerate}{1}
\setlist[well-isol,1]{
  label=\protect\itshape({Isol\;}\arabic*),
  labelindent=0pt,
  leftmargin=*
}
\crefname{well-isoli}{\unskip}{\unskip}
\crefformat{well-isoli}{#2#1#3}
\crefrangeformat{well-isoli}{#3#1#4{\itshape--}#5#2#6}
\Crefrangeformat{well-isoli}{#3(#1)#4 to #5(#2)#6}

\newlist{point-types}{enumerate}{1}
\setlist[point-types,1]{
  label=\protect\itshape({Case\;}\arabic*),
  labelindent=0pt,
  leftmargin=*
}

\newlist{problems}{enumerate}{1}
\setlist[problems,1]{
  label=\protect(P\arabic*),
  labelindent=0pt,
  leftmargin=*
}

\usepackage{marginnote,xcolor}
\newcommand{\marker}{\textcolor{red}{%
    \clap{\rule[-0.9ex]{0.2ex}{2.9ex}}%
    \clap{\rule[-1ex]{1.5ex}{0.2ex}}%
  }}
\newcommand{\marginnotefrom}[2]{\marker\marginnote{\normalfont\scriptsize\textcolor{red}{\textbf{#1 says:}\\#2}}}
\newcommand{\Alex}[1]{\marginnotefrom{Alex}{#1}}

\let\oldabstractname\abstractname
\def\abstractname{\oldabstractname\phantomsection\addcontentsline{toc}{section}{\oldabstractname}}

\title{On the Complexity of Computing with\\ Planar Algebraic Curves}

\usepackage[noblocks]{authblk}

\author[1--3]{Alexander Kobel}
\author[1]{Michael Sagraloff}
\affil[1]{Max-Planck-Institut f\"ur Informatik}
\affil[2]{International Max Planck Research School for Computer Science}
\affil[3]{Universit\"at des Saarlandes\vspace{1ex}\authorcr
  Saarbr\"ucken, Germany\vspace{1ex}\authorcr
  \texttt{\{alexander.kobel,michael.sagraloff\}@mpi-inf.mpg.de}}

\begin{document}

\newgeometry{hmargin=(\paperwidth - \textwidth) / 2}
\maketitle

\begin{abstract}
  In this paper, we give improved bounds for the computational complexity of computing with planar algebraic curves. More specifically, for arbitrary coprime polynomials \(f\), \(g\in\ZZ[x,y]\) and an arbitrary polynomial \(h \in \ZZ[x,y]\), each of total degree less than \(n\) and with integer coefficients of absolute value less than \(2^\tau\), we show that each of the following problems can be solved in a deterministic way with a number of bit operations bounded by \(\tilde{O}(n^6+n^5\tau)\), where we ignore polylogarithmic factors in \(n\) and \(\tau\):
  \begin{itemize}
  \item \emph{The computation of isolating regions in \(\CC^2\) for all complex solutions} of the system \(f = g = 0\),
  \item \emph{the computation of a separating form} for the solutions of \(f=g=0\),
  \item \emph{the computation of the sign of \(h\)} at all real valued solutions of \(f = g = 0\), and
  \item \emph{the computation of the topology} of the planar algebraic curve \(\mathcal{C}\) defined as the real valued vanishing set of the polynomial~\(f\).
  \end{itemize}
  Our bound improves upon the best currently known bounds for the first three problems by a factor of \(n^2\) or more and closes the gap to the state-of-the-art randomized complexity for the last problem.

  \begin{keywords}
    polynomial system solving,
    bivariate systems,
    separating form,
    topology analysis,
    algebraic curves,
    arrangement computation,
    cylindrical algebraic decomposition,
    complexity analysis,
    multipoint evaluation,
    approximate computation
  \end{keywords}
\end{abstract}

\clearpage
\restoregeometry

\section{Introduction}\label{sec:introduction}

In this paper, we derive record bounds for the computational complexity of the following problems, which are related to the arrangement computation of planar algebraic curves:
\begin{problems}
  \setlength\abovedisplayskip{0.25\abovedisplayskip}
  \setlength\belowdisplayskip{0.25\belowdisplayskip}
\item Given two coprime polynomials \(f\), \(g\in\ZZ[x,y]\) of degree \(n\) or less, compute isolating regions in \(\CC^2\) for \emph{all distinct complex} solutions \((x_i,y_i)\in\CC^2\) of the system
\begin{align}\label{system}
f(x,y)=g(x,y)=0,
\end{align}
with \(i=1,\ldots,r\) and some integer \(r\) with \(r\le \deg f\cdot \deg g\le n^2\), which is the upper bound of the number of solutions of the zero-dimensional system due to Bézout's Theorem.
\item Compute a \emph{separating form} \(x+s\cdot y\), with \(s\in\{0,1,\ldots,n^4\}\), for \cref{system} such that \(x_i+s\cdot y_i\neq x_j+s\cdot y_j\) for all \(i,j\) with \(i\neq j\).
\item Given an arbitrary polynomial \(h\in\ZZ[x,y]\), evaluate the sign of \(h\) at all \emph{real valued} solutions \((x_i,y_i)\) of \cref{system}.
\item Given an arbitrary polynomial \(f\in\ZZ[x,y]\), compute the topology of the real planar algebraic curve
\begin{align}\label{curve}
\mathcal{C}:=\{(x,y)\in\RR^2:f(x,y)=0\}
\end{align}
in terms of a planar straight line graph that is isotopic\footnote{We consider the stronger notion of an \emph{ambient isotopy}, but omit the word ``ambient''. A graph \(\mathcal{G}_C\), embedded in \(\RR^2\), is ambient isotopic to \(C\) if there exists a continuous mapping \(\phi:[0,1]\times \RR^{2}\mapsto \RR^2\) with
\(\phi(0,\cdot)=\operatorname{id}_{\RR^{2}}\), \(\phi(1,C)=\mathcal{G}_C\), and \(\phi(t_0,\cdot):\RR^{2}\mapsto
\RR^{2}\) is a homeomorphism for each \(t_0\in[0,1]\).} to \(\mathcal{C}\).
\end{problems}

We remark that a solution to the above problems allows us to answer all necessary queries for arrangement computations with
planar algebraic curves~\cite{bhk-ak2-2011,BEKS13,eigenwilligk08,halperinarrangements,mkerber-diss-09}. Namely, for a set of planar algebraic curves, we can compute the topology of each of these
curves from (P4), we can compute the intersection points of two curves from (P1), and, from (P3), we can decide whether two intersection points from two distinct pairs of curves are equal or not.

The main contribution of this paper is a \emph{deterministic} algorithm that solves all of the above problems (P1) to (P4)
in a number of bit operations bounded by \(\tilde{O}(n^6+n^5\tau)\),
where \(n\) is an upper bound for the total degree of the polynomials \(f\), \(g\) and \(h\), and \(\tau\) is an upper bound for the bitsize of their coefficients.
For the first two problems, we also give more general bounds that take into account the case of unbalanced input.
That is, if \(f\) and \(g\) are polynomials of total degree \(m\) and \(n\) and with integer coefficients of bitsize bounded by \(\tau_f\) and \(\tau_g\), respectively, then (P1) and (P2) can be solved in time \(\sOh \big({\max}^2\set{m,n} \cdot (m^2n^2 + mn (m\tau_g + n\tau_f))\big)\).

We briefly outline our approach:
For (P1), we first extend (and modify) an algorithm from Berberich et al.~\cite{BEKS13,DBLP:conf/alenex/BerberichES11}, denoted \bisolve, that isolates only
the \emph{real valued} solutions of \cref{system}. The so-obtained algorithm \cbisolve computes isolating polydisks in
\(\CC^2\) for all complex solutions and further refines these disks to an arbitrarily small size if necessary.
From a high-level perspective, the algorithm decomposes into two steps: In the first step, the \emph{projection step}, we project all solutions onto
their \(x\)- and \(y\)-coordinates using resultant computation and univariate root finding. This induces a grid consisting of \(O(n^4)\) candidates that have to be checked for solutions in the second step, the \emph{validation step}: For processing the candidates, we combine approximate evaluation of the input polynomials at the candidates and adaptive evaluation bounds derived from the co-factor representations of the resultant polynomials; see \cref{cbisolve-review} for details. We further remark that \cbisolve does not need any coordinate transformation and returns isolating polydisks in the initial coordinate system.

From the solutions of the system \cref{system}, we derive a corresponding separating form \(x+s\cdot y\) for \cref{system} by approximating all "bad" values for \(s\) (to an error of \(1/2\)), for which a pair of distinct solutions is mapped to the same value via \(x+s\cdot y\). Since there are at most \(r\), with \(r\le n^2\), many distinct solutions, there exist at most \(\smash{\binom{r}{2}<n^4}\) bad values for \(s\). Hence, we can determine a separating form with an \(s\in\{0,1,\ldots,\cramped{n^4}\}\). This solves Problem (P2).

For Problem (P3), we first isolate the solutions of the two systems \(f=g=0\) and \(g=h=0\) and, then, determine all common solutions by comparing the corresponding coordinate values. For the sign evaluation of \(h\) at the remaining solutions of \(f=g=0\), we again use approximate evaluation.

For Problem (P4), we first perform a coordinate transformation that
ensures that the curve \(\mathcal{C}\) is in generic situation with respect to the critical points\footnote{A point  \((x,y)\in\CC^2\) is critical for the function \(f\) if
\(\frac{\partial f}{\partial x}(x,y)=\frac{\partial f}{\partial y}(x,y)=0\).} of its defining polynomial \(f\), that is, there exist no two critical points
of \(f\) sharing the same \(x\)-coordinate. In particular, this guarantees that, for any \(\alpha\in\RR\), there exists at most one singular point of \(\mathcal{C}\) located above \(\alpha\).
We remark that this approach crucially differs from previous approaches that
require a coordinate transformation that is generic with respect to all \(x\)-critical points\footnote{A point \((x,y)\in\CC^2\) is an \(x\)-critical point of the curve \(\mathcal{C}\) if \(\frac{\partial f}{\partial x}(x,y)=f(x,y)=0\).} of the curve.
In a second step, we compute all \emph{real valued} \(x\)-critical points \((x_i,y_i)\in\RR^2\) of \(\mathcal{C}\), with \(i=1,\ldots,k\) and some \(k\in\NN\), and evaluate the sign of \(f_x:=\frac{\partial f}{\partial x}\) at all points \((x_i,y_i)\). Based on the latter computation, we can use Teissier's Lemma~\cite{BEKS13,Gwozdziewicz00formulaefor,Teissier} in order to compute the number of distinct (complex) roots of each of the polynomials \(f(x_i,y)\in\RR[x]\); see \cref{sec:topology} for details.
Given the latter information, we can further isolate the roots of the polynomials \(f(x_i,y)\) using a certified numerical method~\cite{DBLP:conf/issac/MehlhornSW13} that works with approximations of the polynomials \(f(x_i,y)\) only.
We further choose values \(\gamma_i\in\RR\) that separate the \emph{\(x\)-critical values} \(x_i\) and isolate the roots of the square-free polynomials \(f(\gamma_i,y)\). Eventually, we connect the ``lifted points'' (i.e.~points \((x,y)\in\mathcal{C}\), where \(x=x_i\) or \(x=\gamma_i\)) via line segments in an appropriate manner. We remark that this last step can be achieved purely based on combinatorial decisions as we can derive the local topology at all non-singular \(x\)-critical points, and there are no two singular points of the curve located above each other; see \cref{sec:topology} for details.

\newparagraph
The analysis of \cbisolve is based on and improves the results by Emeliyanenko and Sagraloff from~\cite{DBLP:conf/issac/EmeliyanenkoS12}.
In particular, the worst case bit complexity of \cbisolve improves by a factor of \(n^2\),
due to the use of asymptotically fast methods for approximate multipoint evaluation~\cite{DBLP:journals/jc/Kirrinnis98,KS13-multipoint} and univariate root
finding~\cite{DBLP:conf/issac/MehlhornSW13,Pan02,SM-ANEWDSC13}, and also due to a slight, but crucial, modification of \bisolve in the validation step.
For the analysis of our algorithms for solving a bivariate system, for computing a corresponding separating form,
and for the sign evaluation of a polynomial at the real-valued solutions, we heavily utilize amortization over the set of all solutions.
For instance, our analysis shows that the cost for computing the sign of a polynomial \(h\) at only one solution of \cref{system} can, in the worst case, be of the same order as the cost for computing the sign of \(h\) at all solutions.
For analyzing our algorithm for topology computation, we use amortized complexity bounds as provided in~\cite{DBLP:conf/issac/MehlhornSW13,DBLP:journals/jsc/KerberS12}. In fact, we consider our algorithm as proposed in the paper at hand to be a deterministic variant of the randomized method introduced in~\cite{BEKS13} and analyzed in~\cite{DBLP:conf/issac/MehlhornSW13}. Here, we remark that the key idea that allows us to derandomize the algorithm from~\cite{BEKS13} is to consider a coordinate transformation that is generic with respect to the critical points of \(f\) and not (necessarily) with respect to the \(x\)-critical points of \(\mathcal{C}\). Following this approach, we can efficiently solve three major problems with respect to topology computation, that is, the computation of a generic coordinate transformation, the computation of the
number of roots of the polynomials \(f(x_i,y)\), where \(x_i\) is an \(x\)-critical value, and the computation of the adjacencies between neighboring fibers in the final step of our algorithm.

\paragraph{Related work}
\addcontentsline{toc}{subsection}{Related work}

Computing the solutions of a polynomial system and computing the topology of an algebraic curve are fundamental problems in computational algebra, and both problems have been well-studied in the past decades. In particular, the literature on polynomial system solving is extensive, and thus, we can only provide a very brief overview, where we refer to an (incomplete) list of relevant books and papers and the references therein.
We mainly distinguish between two classes of approaches for solving polynomial systems, namely, numerical and symbolic methods.

\newparagraph
\emph{Numerical algorithms} are predominantly used in engineering applications, mainly because of their high efficiency for real-world instances, their applicability for high-dimensional systems, and their relative ease of implementation.
However, the major drawback of most numerical methods is that they are prone to fail in degenerate cases such as singular or clustered solutions, or if the solution set is infinite.
Also, they typically do not provide any guarantees about correctness or completeness of their output.
Popular representatives of the class of numerical methods are homotopy solvers such as \textsc{Hom4PS} by Li and Tsai \cite{Li2003209}, \textsc{Bertini} by Bates et al.~\cite{sommese2005numerical},
or \textsc{PHCpack} by Verschelde \cite{Verschelde:2011:PHC:1940475.1940524}.
Their common modus operandi is to heuristically track the known solution set of a rather simple polynomial system under continuous transformations until the problem at hand is reached; e.g., consider~\cite{sommese2005numerical,bates2013numerically} for a great introduction into the field.
Sophisticated strategies to ensure robustness against numerical instabilities have been developed, and recent work even introduces techniques to certify the computed solutions as so-called \emph{approximate solutions}; e.g., see~\cite{bates2013numerically,DBLP:journals/focm/BeltranL13} and the references given therein. However, although such methods can be used to compute arbitrarily good approximations of the solutions of a polynomial system, they must fail to isolate multiple roots due to the single use of approximate arithmetic. For polynomial systems with integer coefficients, a possible solution to the latter problem is to compute very good approximations of the solutions and, then, to use worst-case separation bounds for the distance between two distinct solutions. However, this approach requires extremely good approximations of the solutions for any given input, and thus, is considered to be irrelevant in practice. In addition, we are not aware of worst-case (bit) complexity bounds for such an approach that are comparable to the paper at hand.

Another kind of numerical methods is based on subdivision schemes; e.g.,~\cite{mourrain-sub,mp-subdivision,DBLP:journals/jsc/BurrCGY12,Plantinga:2004:IAI:1057432.1057465}.
Those algorithms work on a region of interest, which can a priory be chosen large enough to comprise all global solutions of the input.
The region is then recursively subdivided into smaller regions until, for each of them, either the size of the region is smaller than some predefined threshold or it can be ensured that the region contains no or exactly one unique solution.
For instance, this can be achieved using predicates based on interval arithmetic or Rouch\'e's criterion.
An important benefit is the locality of the approach, however, it requires efficient predicates for excluding empty cells to confine an exponential growth of the recursion tree, in particular in higher dimensions.
In fact, for more than one dimension, we are not aware of any polynomial-time subdivision method for approximating all solutions that uses only purely numerical predicates.
In addition, for isolating multiple solutions, subdivision methods share the same difficulties with other purely numerical methods. Again, for integer or rational input, worst-case separation bounds can be used to handle even degenerate cases, however, for the price of rendering the overall method impractical; e.g., see~\cite{DBLP:journals/jsc/BurrCGY12}.

\newparagraph
The second main approach, which is also considered in the paper at hand, rather uses \emph{symbolic} than numerical methods for solving
polynomial systems. That is, in a first step, they use elimination techniques such as the computation of resultants or Gr\"obner Bases to
reduce the number of unknowns. For zero-dimensional systems, this means that the solutions are projected into one dimension via algebraic manipulations, followed by a
computation of the projections as the roots of the corresponding elimination polynomial. Then, in a second step, the solutions are recovered
from their projections. The latter (lifting) step is relatively cheap if the projection already comes along with a parametrization of the solutions as provided by a rational univariate representation (RUR); e.g., see~\cite{Gonzalez-Vega:1996:IUC:247153.247185,Diochnos:2009:APC:1530888.1530918,Bouzidi:2013:SLF:2465506.2465518,r-rsicms-2010}. In contrast, the lifting step can be quite costly if the projection is not known to separate the solutions as given, at least in general, for a triangular decomposition. Computing the topology of a real algebraic plane curve \(\mathcal{C}=\{(x,y)\in\RR^2:f(x,y)=0\}\) can be achieved using Collin's classical cylindrical algebraic decomposition technique~\cite{c-qercfcad-75}, dating back to 1975; see also~\cite{bpr-algorithms}.
It heavily relies on the computation of the \(x\)-critical points of the curve, and thus, on the computation of the solutions of the bivariate polynomial system \(f=f_y=0\). In this context, it is remarkable that all known complete and certified algorithms, for which polynomial time bounds are known, compute (more or less) a cylindrical algebraic decomposition.

When compared to purely numerical methods, the main advantage of symbolic methods is that they can be made complete and certified, even in the presence of a degenerate situation such as multiple solutions of a polynomial system or singular points of an algebraic curve. However, this also comes with a price, namely, the high computational cost for the considered symbolic operations, which do not adapt to the actual hardness of the input. As a consequence, when solving a polynomial system, the running time mainly depends on the considered symbolic operations, and not on the geometry of the solutions. That is, it makes almost no difference whether the solutions are simple and well separated from each other or whether there exist clustered or multiple solutions. However, a considerable amount of work has been invested in the past years to reduce the number of purely symbolic operations and to replace algebraic manipulations by approximate arithmetic, without abstaining from completeness or correctness; e.g., see~\cite{Strzebonski,CLPPRT10,BEKS13,mkerber-diss-09,LGP-09} for an (incomplete) list of recent papers that describe hybrid methods combining the computation of a cylindrical algebraic decomposition and numerical computation.
As a consequence, the bounds on the theoretical complexity for the problems (P1) and (P4) have been improved over the years in an impressive manner:
Arnon and McCallum \cite{Arnon1988213} gave the first sub-exponential bit complexity bound \(\tilde{O}(n^{30} + n^{27}\tau^3)\) for deterministically computing the topology of a planar algebraic curve, and this bound was subsequently improved\footnote{The references are sorted with respect to the size of the given bound for the bit complexity of the corresponding algorithm.} to \(\tilde{O}(n^9\tau+n^8\tau^2)\) by Cheng et al.~\cite{CLPPRT10}, Gonzalez-Vega and El Kahoui~\cite{Gonzalez-Vega:1996:IUC:247153.247185}, Basu et al.~\cite[Sec.~11.6]{bpr-algorithms}, Diochnos et al.~\cite{Diochnos:2009:APC:1530888.1530918}, and Kerber and Sagraloff~\cite{mkerber-diss-09,DBLP:journals/jsc/KerberS12}. A very recent manuscript\footnote{According to personal communication (January 2014) with Marie-Fran\c{c}oise Roy, who sent us a corresponding preprint of~\cite{hal.archives-ouvertes.fr:hal-00935728}.} even reports on a deterministic algorithm for computing the topology of an algebraic plane curve that uses \(\tilde{O}(n^7+n^6\tau)\) bit operations, and thus, lags behind our method by only one magnitude. For randomized methods, the current record bound~\cite{DBLP:conf/issac/MehlhornSW13} on the expected number of bit operations is even as low as \(\tilde{O}(n^6+n^5\tau)\), and thus, comparable to our bound.

For the special task of solving a bivariate polynomial systems, several projection based methods have been recently presented and analyzed. In~\cite{Diochnos:2009:APC:1530888.1530918}, Diochnos et al.\ discuss three methods to solve a bivariate polynomial system. The first method, \textsc{Grid}, is similar to our method \cbisolve and its predecessor \bisolve in the sense that, in a first step, it determines a grid of candidate solutions via projecting the solutions onto the \(x\)- and \(y\)-axis.
However, verification of the candidates is done in a completely different manner, that is, candidates are either verified as solutions or discarded by evaluating a corresponding subresultant sequence at the endpoints of isolating intervals for the projected solutions.
The bit complexity of \textsc{Grid} is bounded by \(\tilde{O}(n^{14} + n^{12}\tau^2)\), where the overall cost is dominated by that of the evaluation of the subresultant sequence.
The second approach, \textsc{M\_rur}, is based on the computation of an RUR and achieves a bit complexity of \(\tilde{O}(n^{12}+n^{10}\tau^2)\).
The third approach, \textsc{G\_rur}, computes the greatest common divisor of the square-free parts of \(f(\alpha,y)\) and \(g(\alpha,y)\), where \(\alpha\) is the projection of a solution of the system. Its bit complexity is also bounded by \(\tilde{O}(n^{12}+n^{10}\tau^2)\). In~\cite{LGP-09}, Cheng et al.\ propose the so-called local-generic position method (LGP for short). Instead of considering a coordinate transformation that is separating for all solutions, they consider, for each projection of a solution, a corresponding transformation that is locally separating for all the solutions that are located above a small neighborhood of the projection. The method shows good efficiency in practice, however, with respect to worst-case bit complexity, it suffers from the fact that such a separating form may be of a large bit-size. It is analyzed in~\cite{DBLP:journals/corr/ChengJ13}, where the authors derive the bound \(\tilde{O}(n^{10} + n^9\tau)\) for its bit complexity.
The best current deterministic algorithms for solving bivariate systems are due to Bouzidi et al.~\cite{Bouzidi:2013:SLF:2465506.2465518,Bouzidi:2013:RUR:2465506.2465519} and Emeliyanenko et al.~\cite{DBLP:conf/alenex/BerberichES11,DBLP:conf/issac/EmeliyanenkoS12}. It has been shown that both methods isolate the solutions of a bivariate system with \(\tilde{O}(n^8+n^7\tau)\) bit operations. The method from Bouzidi et al.\ is based on the computation of a separating form followed by the computation of a RUR of the system. For both steps, the authors derive the bound \(\tilde{O}(n^8+n^7\tau)\) for the bit complexity. Very recent work\footnote{According to personal communication (January 2014) with Sylvain Lazard, who sent us a preliminary version of~\cite{bouzidi:hal-00977671} stating a deterministic worst-case bound of size \(\tilde{O}(n^7+n^6\tau)\) for the computation of a separating form for bivariate polynomial systems.} indicates that the latter bound can be even lowered by a factor \(n\). For randomized algorithms, the current record bound~\cite{DBLP:conf/issac/MehlhornSW13} on the expected number of bit operations is \(\tilde{O}(n^6+n^5\tau)\), however, the algorithm can only be used to compute the real solutions of the system \(f=g=0\) as it reduces the latter problem to the problem of computing the topology of the real algebraic curve defined by \(f^2+g^2=0\).

\pagebreak[4]
In the paper at hand, we extend and modify the algorithm \bisolve as presented in~\cite{DBLP:conf/alenex/BerberichES11} to complex system solving, and we considerably improve upon the analysis from~\cite{DBLP:conf/issac/EmeliyanenkoS12}. We further show that we can solve the problems (P2) and (P3) with \(\tilde{O}(n^6+n^5\tau)\) bit operations, given that isolating regions for all complex solutions are already computed. Together with our bound for the bit complexity of \cbisolve, this yields an improvement upon the current record bounds~\cite{Bouzidi:2013:SLF:2465506.2465518,Bouzidi:2013:RUR:2465506.2465519} for the latter two problems by a factor \(n^2\) and a factor \(n^3\), respectively. Finally, we combine \cbisolve with the algorithm \textsc{TopNT} from~\cite{BEKS13} to derive our fast method for computing the topology of an algebraic plane curve.

\section*{Notation}
\label{sec:notation}

Throughout the paper, we use the following notations:
\begin{itemize}
\item Complementary to the common Landau notation \(\Oh(\cdot)\) for asymptotic behavior, \(\sOh(\cdot)\) means that we ignore poly-logarithmic factors, that is, \(\tilde{O}(T) = O(T (\log T)^k)\), where \(k\) is any fixed integer.
\item For \emph{disks in the complex plane,} we write \(D_r(m) := \set{ z \in \CC : \abs{z-m} < r }\), where \(m \in \CC\) denotes the center and \(r \in \RR_{>0}\) the radius.
\item For arbitrary complex values \(w,z \in \CC\), we define \(M(z) \coloneqq \max \set{1, \abs{z}}\) and \(M(z,w) \coloneqq \max \set{1, \abs{z}, \abs{w}}\).
\item We write \(\log \coloneqq \log_2\) for the binary logarithm,
  and define \(\LOG z \coloneqq \lceil M(\log M(z)) \rceil\).
  That is, if \(\abs{z} \le 2\), then \(\LOG z\) is \(1\); otherwise, \(\LOG z\) equals \(\log \,\abs{z}\) rounded up to the next integer.
\item The \emph{bitsize} \(\tau(a) \coloneqq \LOG(\abs{a}+1)\) of an integer \(a \in \ZZ\) is the length of its binary representation, neglecting the sign bit.
  Notice that \(\tau(a) \ge 1\) by definition. For an arbitrary complex number \(\frac{a}{b}+\frac{c}{d}\cdot\ii\) with integers \(a\) to \(d\), we define its bitsize as the maximum of the values \(\tau(a)\) to \(\tau(d)\).
\item A polynomial \(F \in \ZZ[x_1,\ldots,x_k]\) is of \emph{magnitude} \((n, \tau)\) if its total degree and the bitsize of its coefficients are upper bounded by \(n\) and \(\tau\), respectively.
\item\label{def:reverse-polynomial} For a polynomial \(F \in \CC[x]\) and an integer \(d\) with \(d \ge \deg F\), we write \(\rev_d F \coloneqq x^d\cdot F(1/x)\) for the \emph{reverse of \(F \in \CC[x]\) with respect to \(x^d\),}
  that is, the polynomial with reversed coefficient sequence of that of \(F\) considered as a \(d\)-th degree polynomial.
\item For a polynomial \(F \in \CC[x]\) with pairwise distinct roots \(z_1, \dots, z_k \in \CC\), we define
  \begin{itemize}
  \item \(F^\ast \coloneqq F / \gcd(F,F')\),\quad the \emph{square-free part of \(F\),}
  \item \(\mult(z_i, F) \coloneqq \min \set{ m \in \NN : \frac{\partial^m}{\partial x^m}F(z_i)\neq 0 }\),\quad the \emph{multiplicity of \(z_i\)}
  \item \(\separ(z_i, F) \coloneqq \min_{j \ne i} \abs{z_i - z_j}\),\quad the \emph{separation of \(z_i\)}
  \item \(\separ(F) \coloneqq \min_i \separ(z_i,F)\),\quad the \emph{root separation of \(F\),}
  \item \(\Mea(F) \coloneqq \abs{\lcf(F)} \cdot \prod_{i=1}^k M(z_i)^{\mult(z_i,F)}\),\quad the \emph{Mahler measure of \(F\),}
  \item \(\Gamma(F) := \max_i \abs{z_i}\),\quad the \emph{maximum modulus of the roots of \(F\),}
  \item \(\Sigma(F) := \sum_{i=1}^k \mult(z_i,f) \cdot \LOG (\separ(z_i,F)^{-1})\),\quad and
  \item \(\Sigma^\ast(F) := \Sigma(F^\ast) = \sum_{i=1}^k \LOG (\separ(z_i,F)^{-1})\).
  \end{itemize}
  If it is clear from the context, we omit \(F\) and simply write \(\mult_i\) or \(\mult_{z_i}\) and \(\separ_i\) or \(\separ_{z_i}\) in place of \(\mult(z_i,F)\) and \(\separ(z_i,F)\).
\item For a square matrix \(A = (a_{ij}) \in \CC^{n\times n}\), we write \(\had(A)\) for the row-wise Hadamard's bound on the determinant of \(A\):
  \begin{align*}
    \had(A) \coloneqq \Big(\prod_{i=1}^n \sum_{j=1}^n \abs{a_{ij}}^2 \Big)^{1/2}.
  \end{align*}
\end{itemize}

\addcontentsline{toc}{section}{Notation}
\section{Solving a Bivariate Polynomial System}

We first consider the problem of computing isolating regions in \(\CC^2\) for all complex solutions of a system of equations
\begin{align}
  \label{eq:bisolve-system}
  f(x,y) = \sum_{i+j \le m} f_{i,j} \, x^iy^j = 0
  \qquad\text{and}\qquad
  g(x,y) = \sum_{i+j \le n} g_{i,j} \, x^iy^j = 0,
\end{align}
where \(f\), \(g \in \ZZ[x,y]\) are bivariate polynomials of total degrees \(m\ge 1\) and \(n \ge 1\), respectively.
Furthermore, let the absolute values \(\abs{f_{i,j}}\) and \(\abs{g_{i,j}}\) of the coefficients of \(f\) and \(g\) be bounded by \(2^{\tau_f}\) and \(2^{\tau_g}\) for some positive integers \(\tau_f\), \(\tau_g \ge 1\).
We assume that the solution set
\begin{align}
  \label{eq:solution-set}
  V \coloneqq \set{ (\alpha,\beta) \in \CC^2 : f(\alpha,\beta) = g(\alpha,\beta) = 0 }
\end{align}
is \emph{zero-dimensional} or, equivalently, that \(f\) and \(g\) share no common nontrivial factor in \(\CC[x,y] \setminus \CC\).
Then, B\'ezout's theorem states that the number of solutions counted with multiplicity is at most \(m\cdot n\). In fact, the number of solutions equals \(m\cdot n\) if we consider the corresponding homogeneous system and also count the solutions at infinity.

We assume that the input polynomials are given by their coefficients.
Considering \(f\) and \(g\) as \emph{univariate} polynomials in \(x\) and \(y\), respectively, we write:
\begin{align}
  \label{eq:bisolve-system-considered-univariate}
  \begin{split}
    f(x,y) = \sum_{i=0}^{m_x} f_i^{(x)}(y) \, x^i &= \sum_{i=0}^{m_y} f_i^{(y)}(x) \, y^i
    \qquad\text{and}\\
    g(x,y) = \sum_{i=0}^{n_x} g_i^{(x)}(y) \, x^i &= \sum_{i=0}^{n_y} g_i^{(y)}(x) \, y^i,
  \end{split}
\end{align}
where \(f_{i}^{(y)}\), \(g_{i}^{(y)} \in \ZZ[x]\) and \(f_{i}^{(x)}\), \(g_{i}^{(x)} \in \ZZ[y]\), and \(m_x\), \(m_y \le m\) and \(n_x\), \(n_y \le n\) denote the degrees of \(f\) and \(g\) as univariate polynomials.

\subsection{Some facts about resultants}\label{subsec:basics}

Our approach for handling bivariate systems is based on elimination of variables by computing resultants.
Before we present our algorithm to isolate the solutions of \cref{eq:bisolve-system}, we briefly review some fundamental and well known results about resultants.
The reader familiar with the topic may skip this part and will find pointers to the explanations in the main text.
We try to keep our description as short as possible and mainly focus on our applications.
For an in-depth discussion of the subject, there is an extensive range of literature for further study; e.g., we recommend \cite{bpr-algorithms,GKZ94,vzGG13,yap00}.
\begin{definition}[Sylvester matrix]
  \label{def:sylvester-matrix}
  For two univariate polynomials \(f = \sum_{i=0}^m f_ix^i\) and \(g = \sum_{i=0}^n g_ix^i \in \DD[x]\) of degree \(m = \deg f\) and \(n = \deg g\) over some integral domain \(\DD\),
  with \(f\not\equiv 0\) and \(g \not\equiv 0\),
  their \emph{Sylvester matrix \(S(f,g; x) \in \DD^{(m+n)\times(n+m)}\) with respect to \(x\)} is defined as
  \begin{align*}
    S(f,g; x) \coloneqq \scalebox{0.9}{\(
      \begin{pmatrix}
        f_{m} & f_{m-1} & \cdots  & f_{0}                             \\[0.5ex]
        & f_{m}   & f_{m-1} & \cdots & f_{0}                    \\
        &         & \ddots  &        &         & \ddots         \\
        &         &         & f_{m}  & f_{m-1} & \cdots & f_{0} \\[0.5ex]
        g_{n} & g_{n-1} & \cdots  & g_{0}                             \\[0.5ex]
        & g_{n}   & g_{n-1} & \cdots & g_{0}                    \\
        &         & \ddots  &        &         & \ddots         \\
        &         &         & g_{n}  & g_{n-1} & \cdots & g_{0} \\
      \end{pmatrix}\)
    }
    \!.
  \end{align*}
  Notice that there are \(n\) rows with coefficients of \(f\) and \(m\) rows with coefficients of \(g\).

  In our application, \(f\), \(g \in \ZZ[x,y]\), and we consider both Sylvester matrices with respect to \(y\) or \(x\) as the outer variable.
  Hence, we write \(S(f,g; y)\) for \(\DD=\ZZ[x]\) (and the entries of the Sylvester matrix are polynomials in \(\ZZ[x]\)), and \(S(f,g; x)\) for \(\DD = \ZZ[y]\).
\end{definition}

\begin{definition}[Resultant]
  For \(f\) and \(g\) as in \cref{def:sylvester-matrix}, the \emph{resultant \(\res(f,g; x)\) of \(f\) and \(g\) with respect to \(x\)} is the determinant of \(S(f,g; x)\).
  It will be convenient to additionally define \(\res(0,0; x) \coloneqq 1\) and \(\res(f,0; x) = \res(0,g; x) \coloneqq 0\) for all \(f\), \(g \not\equiv 0\).
\end{definition}

\begin{theorem}[Properties of resultants]
  \label{res-prop}
  Let \(f\) and \(g\) be polynomials in \(\DD[x]\) as above.
  The resultant of \(f\) and \(g\) satisfies the following properties:
  \begin{enumerate}
  \item \label{res-prop:root-distance}
    If \(f(x) = \lcf f \cdot\prod_{i=1}^m (x-x_i)\) and \(g(x) = \lcf g \cdot\prod_{j=1}^n (x-y_j)\), where \(x_i\) and \(y_i\) are the roots of \(f\) and \(g\) in an algebraic closure \(\overline{\DD}\) of \(\DD\),
    possibly repeated according to their multiplicities, then
    \begin{align*}
      \res(f,g;x) &= (\lcf f)^n \prod\nolimits_i g(x_i) \quad=\quad (\lcf g)^m \prod\nolimits_j f(y_j)\\
      &= (\lcf f)^n (\lcf g)^m \prod\nolimits_{i,j} (x_i - y_j).
    \end{align*}

  \item \label{res-prop:gcd}
    \(\res(f,g;x) \equiv 0\) if and only if \(f\) and \(g\) share a common nontrivial factor in \(\DD[x] \setminus \DD\).

  \item \label{res-prop:cofactor-rep}
    The resultant can be represented as a \(\DD[x]\)-linear combination of \(f\) and \(g\) as
    \begin{align*}
      \res(f,g;x) = u(x) \cdot f(x) + v(x) \cdot g(x),
    \end{align*}
    where \(u\in\DD[x]\) and \(v \in \DD[x]\) are polynomials of degree less than or equal to \(\deg g\) and \(\deg f\), respectively. Furthermore, the polynomials \(u\) and \(v\) can be written as the determinants of the ``Sylvester-like'' matrices
    \begin{align*}
      U(x) =\! \scalebox{0.9}{\(
        \begin{pmatrix}
          f_m & \cdots & f_0 &        & \hspace{-1em}x^{n-1} \\[-1ex]
          & \ddots &     & \ddots & \hspace{-1em}\vdots  \\
          &        & f_m & \cdots & \hspace{-1em}1       \\
          g_n & \cdots & g_0 &        & \hspace{-1em}0       \\[-1ex]
          & \ddots &     & \ddots & \hspace{-1em}\vdots  \\
          &        & g_n & \cdots & \hspace{-1em}0
        \end{pmatrix}\)
      }
      \ \ \!\!\text{and}\!\!\quad
      V(x) =\! \scalebox{0.9}{\(
        \begin{pmatrix}
          f_m & \cdots & f_0 &        & \hspace{-1em}0       \\[-1ex]
          & \ddots &     & \ddots & \hspace{-1em}\vdots  \\
          &        & f_m & \cdots & \hspace{-1em}0       \\
          g_n & \cdots & g_0 &        & \hspace{-1em}x^{m-1} \\[-1ex]
          & \ddots &     & \ddots & \hspace{-1em}\vdots  \\
          &        & g_n & \cdots & \hspace{-1em}1
        \end{pmatrix}\)
      }
      \!,
    \end{align*}
    which are obtained from the Sylvester matrix \(S(f,g; x)\) by replacing the last column with \((x^{n-1}, \dots, 1, 0^m)\) and \((0^n, x^{m-1}, \dots, 1)\), respectively.

    \suspend{enumerate}
    Let \(f\) and \(g \in \ZZ[x,y]\) be bivariate polynomials without common nontrivial factor as in~\cref{eq:bisolve-system,eq:bisolve-system-considered-univariate}.
    Consider \(f\) and \(g\) as univariate polynomials over \(\DD = \ZZ[x]\), and write \(R \coloneqq \res(f,g; y)\) for their resultant with respect to \(y\).
    Then, it holds:

    \resume{enumerate}
    \item
    \label{res-prop:degree}
    \(R \in \ZZ[x]\) and \(\deg R\le m\cdot n\).
  \item \label{res-prop:projection}
    The roots of \(R\) are exactly the projections of the solutions of \(f = g= 0\), including points at infinity, onto the complex \(x\)-plane.
    More precisely,
    \begin{align*}
      R(\alpha) = 0
      \quad\text{if and only if}\quad
      \begin{cases}
        f(\alpha,\beta) = g(\alpha,\beta) = 0 \text{ for some } \beta \in \CC&\text{or}\\
        f^{(y)}_{m_y}(\alpha) = g^{(y)}_{n_y}(\alpha) = 0.
      \end{cases}
    \end{align*}
    The multiplicity \(\mult(\alpha,R)\) is the sum of the intersection multiplicities%
    \footnote{The multiplicity of a solution \((\alpha,\beta)\) of~\cref{eq:bisolve-system} is defined as the dimension of the localization of \(\CC[x,y] / \langle f,g\rangle\) at \((\alpha,\beta)\),
      considered as a \(\CC\)-vector space (cf.~\cite[Sec.~4.5, p.~148]{bpr-algorithms}). For a root \(\alpha\) of \(\gcd(f^{(y)}_{m_y}, g^{(y)}_{n_y})\), the intersection multiplicity at the ``infinite point'' \((\alpha,\infty)\) has also to be taken into account, however,
      for simplicity, we decided not to consider the more general projective setting.}
    of all solutions of~\cref{eq:bisolve-system} with \(x\)-coordinate \(\alpha\).
  \end{enumerate}
\end{theorem}

In order to estimate the magnitude of the resultant polynomials, we will use the following generalization of Hadamard's bound for the size of determinants to the case of polynomial entries.

\begin{theorem}[Goldstein-Graham]
  \label{goldstein-graham}
  Let \(A(x) = (a_{ij})\) be an \(n\)-square matrix whose entries \(a_{ij}(x)\) are polynomials in \(\CC[x]\).
  Denote by \(B = (b_{ij})\) the matrix with \(b_{ij} = \sumnorm{a_{ij}(x)}\) and by \(\had(B)\) the row-wise Hadamard bound for its determinant.
  The polynomial determinant \(\det A(x)\) of \(A\) satisfies the inequality
  \begin{align*}
    \norm{\det A(x)}_2 \le \had(B).
  \end{align*}
\end{theorem}
\begin{proof}
  See, e.g.,~\cite[Thm.~6.31]{yap00} or Lossers' elegant solution in \cite{doi:10.1137/1016065} to the claim stated by Goldstein and Graham as an exercise problem.
\end{proof}

\begin{corollary}
  \label{resultant-magnitude}
  If \(f\) and \(g\) as in~\cref{eq:bisolve-system,eq:bisolve-system-considered-univariate} have magnitude \((m, \tau_f)\) and \((n, \tau_g)\), respectively,
  then their resultant \(R \coloneqq \res(f,g;y)\) has magnitude
  \begin{align*}
    (mn, m\tau_g + n\tau_f + \Oh (m\log n + n\log m)).
  \end{align*}
\end{corollary}
\begin{proof}
  The bound on the degree follows immediately from \itemref{res-prop:degree}; it remains to prove the claim on the bitsize.
  Choose \(A = S(f,g;y)\) in \cref{goldstein-graham}.
  The \(1\)-norms of the entries corresponding to \(f\) and \(g\) in the Sylvester matrix are bounded by \((m+1)2^{\tau_f}\) and \((n+1)2^{\tau_g}\).
  Since Hadamard's bound is monotone in the entries of the matrix, we can overestimate the non-zero entries of \(B\) by those values.
  As there are \(n\) rows with \(m_y+1 \le m+1\) entries corresponding to \(f\), and \(m\) rows for \(g\), \cref{goldstein-graham} yields
  \begin{align*}
    \log \norm{R}_2 \le \log \had(B) &= \log \Big( \big((m+1)^{3/2} \, 2^{\tau_f}\big)^n \cdot \big((n+1)^{3/2} \, 2^{\tau_g}\big)^m \Big)\\
    &= n\tau_f + m\tau_g + \Oh(n\log m + m\log n).
    \tag*\qedhere
  \end{align*}
\end{proof}

\begin{lemma}
  \label{resultant-computation-complexity}
  Under the above assumptions on the magnitude of \(f\) and \(g\), the resultant \(R\) can be computed with \(\sOh (\max\set{m,n}^3(m\tau_g + n\tau_f))\) bit operations.
\end{lemma}
\begin{proof}
  We only give the main ideas behind a small primes reduction and lifting approach for \(R = \res(f,g;y)\).
  For an in-depth explanation see, e.g.,~\cite[Sec.~2.4.3 and~2.5.4]{pavel-phd};
  there, Emeliyanenko discusses the symmetric case where \(m = n\) and \(\tau_f = \tau_g\), but the arguments carry over.\footnote{In~\cite{pavel-phd,Emeliyanenko:2013:CRG}, Emeliyanenko reports on an extremely efficient implementation of the small primes and lifting approach on modern graphics hardware. For this, he uses a generalized Schur Algorithm for computing the determinant of a matrix with low displacement rank (as the Sylvester matrix). We remark that his implementation is also integrated in the current C++ implementation of \bisolve \cite{DBLP:conf/alenex/BerberichES11,BEKS13}.}

  Write \(N = mn\) and \(T = \sOh(m\tau_g + n\tau_f)\) for the bounds on the degree and the bitsize of \(R\) as given in \cref{resultant-magnitude}; here, \(T\) is an explicit bound as computed in the proof of~\cref{resultant-magnitude}.
  To recover the coefficients of \(R\) via Chinese remaindering,
  it suffices to compute the images \(R_{p_i} \in \ZZ_{p_i}[x]\) of \(R\) modulo pairwise distinct primes \(p_1, \dots, p_r\) such that \(\prod_i p_i > 2^T\).
  We choose \(r = \Oh(T)\) primes of small magnitude \(\mu(p_i) = \sOh (\log T)\); notice that this is possible due to the prime number theorem, and the cost is bounded by \(\sOh(T)\) since the cost for testing an integer of bitsize \(L\) for being prime is polynomial in \(T\).

  Using asymptotically fast methods \cite[Thm.~10.24]{vzGG13}, the simultaneous reduction of any integer of bitsize bounded by \(\Oh(T)\) over all moduli \(p_1, \dots, p_r\) can be accomplished in time \(\sOh(T)\).
  Since the number of coefficients of \(f\) and \(g\) is quadratic in the total degree, reducing all coefficients requires \(\sOh((m^2+n^2)T) = \sOh(\max\set{m,n}^2\cdot T)\) bit operations.

  The results in \cite{vzGG13,reischert-97,lickteig-01} show that, for two bivariate polynomials in \(\FF[x,y]\) over a field \(\FF\) with total degrees bounded by \(d\),
  their resultant can be computed with \(\sOh(d^3)\) arithmetic operations in \(\FF\).
  Since the bitsizes of the \(p_i\) do not exceed \(\sOh(\log T)\), this translates to \(\sOh (\max\set{m,n}^3 \log T)\) bit operations for each individual \(R_{p_i}\),
  and the total complexity of computing all \(R_{p_i}\) is \(\sOh (\max\set{m,n}^3 \cdot T)\).

  Finally, lifting the \(N\) coefficients of \(R\) from their modular images, again using asymptotically fast methods \cite[Thm.~10.25]{vzGG13}, takes \(\sOh(NT)\) operations.
  Hence, the overall time spent for the computation of \(R\) is dominated by the cost \(\sOh (\max\set{m,n}^3 \cdot T)\) of computing the modular resultants.
\end{proof}

\subsection{Root finding and multipoint evaluation}
\label{known-basic-algorithms}

Based on an asymptotically fast method~\cite{Pan02} for the approximate factorization of a polynomial due to Pan,
recent work~\cite{DBLP:conf/issac/MehlhornSW13} provides an asymptotically fast method for isolating and approximating the roots of an arbitrary polynomial \(F\in\CC[x]\).
As input the algorithm receives an oracle which can deliver arbitrarily good approximations of the coefficients, as well as the number of distinct complex roots.
The following theorem summarizes the complexity results for this method;
see also~\cite{SM-ANEWDSC13} for a dedicated real root solver that achieves a comparable bit complexity bound provided that the input polynomial is square-free.

\begin{theorem}
  \label{univariate-root-isolation}
  \cite[Thm.~3 and Thm.~5]{DBLP:conf/issac/MehlhornSW13}

 \begin{enumerate}
 \item Let \(F\in\CC[x]\) be an arbitrary polynomial of degree \(d\) with complex coefficients of absolute value less than \(2^\mu\) and leading
 coefficient of absolute value larger than \(1\). In addition, suppose that the number \(k\) of distinct roots of \(F\) is given. Then, we
 can compute isolating disks \(D(z_i) = D_{r_i}(m_i)\) with \(r_i < \separ_i / (64d)=\separ(z_i,F)/(64d)\) for all complex roots \(z_1\) to \(z_k\) of
 \(F\), together with their corresponding multiplicities \(\mult_i=\mult(z_i,F)\), using
 		\begin{align*}
				\tilde{O}\Big(d^{3}+d^{2}\tau+d\cdot\sum\nolimits_{i=1}^k \Big(\LOG \big(\tfrac{\partial^{\mult_i}}{\partial x^{\mult_i}} F(z_i)\big)^{-1}+ \mult_{i}\cdot\LOG(\separ_i^{-1})\Big)\Big)
		\end{align*}
		bit operations. For this, we need an approximation of \(F\) to \(\rho\) bits after the binary point, with \(\rho\) bounded by
	\begin{align*}
	\sOh\Big(d\cdot\LOG(\Gamma(F))+\sum\nolimits_{i=1}^k \Big(\LOG \big(\tfrac{\partial^{\mult_i}}{\partial x^{\mult_i}} F(z_i)\big)^{-1}+ \mult_{i}\cdot\LOG(\separ_{i}^{-1})\Big)\Big).
	\end{align*}

 \item Let \(F \in \ZZ[x]\) be a univariate polynomial of magnitude \((d,\mu)\).
  Then, we can compute isolating disks \(D(z_i)\) for all roots \(z_i\) as above together with their corresponding multiplicities \(\mult_i\), using
  a number of bit operations bounded by
  \(
    \sOh(d^3 + d^2\mu)
    .
  \)
  For further refining all the isolating disks to a size of less than \(2^{-L}\), with \(L\) an arbitrary positive integer, we need
  \(
   \sOh(d^3 + d^2\mu+dL)
  \)
  bit operations.
 \end{enumerate}
\end{theorem}

We also need the following well-known bound for the Mahler measure of a polynomial:

\begin{lemma}
  \label{bound-Mea}\cite[Prop.~10.8]{bpr-algorithms}
  Let \(F \in \ZZ[x]\) be a univariate polynomial of magnitude \((d,\mu)\) with pairwise distinct roots \(z_1, \dots, z_m \in \CC\).
  It holds that
  \begin{align}
    \label{eq:bound-Mea}
    \sum\nolimits_{i=1}^m {\mult_i} \cdot \log M(z_i)
    \le \log \Mea(F) \le \log \norm{F}_2= \Oh (\mu + \log d).
\end{align}
\end{lemma}

The following theorem provides a bound for the sum of all values \(\LOG\separ(z_i,F)^{-1}\), where we count each value according to the multiplicity of the corresponding root. Since \cref{bound-Sigma} has already been presented in~\cite{DBLP:conf/issac/EmeliyanenkoS12}, we decided to outsource its proof to \cref{appendix1}:

\begin{theorem}
  \label{bound-Sigma}
  Let \(F\) be as above.
  It holds that
  \begin{align*}
    \sum\nolimits_{i=1}^m \mult(z_i,F)\cdot\LOG\separ(z_i,F)^{-1} = \sOh(d^2+d\mu).
  \end{align*}
  In particular,
  \begin{align}
    \label{eq:bound-Sigma}
    \Sigma^\ast(F)
    &\le \Sigma(F)
    = \sOh(d^2 + d\mu).
  \end{align}
\end{theorem}

Finally, we provide a complexity bound for approximately computing the values of a polynomial \(F\in\CC[x]\) at a set of \(d=\deg F\) many points. The underlying algorithm essentially follows the classical fast multipoint evaluation scheme~\cite{moenckborodin72,vzGG13} by Moenck and Borodin, however, the considered polynomial multiplications and divisions are carried out with approximate but certified arithmetic.
We provide a self-contained proof of \cref{multipoint-evaluation} in \cref{appendix2}.%
\footnote{Kirrinnis~\cite{DBLP:journals/jc/Kirrinnis98} provides a bit complexity bound for the fast multipoint evaluation scheme when using approximate arithmetic. The analysis makes use of a fast method for approximate polynomial division due to Sch\"{o}nhage~\cite{Schoenhage82}. In \cref{appendix2}, we show that, as an alternative, one can also use the classical fast division scheme based on Newton iteration without diminishing the complexity results. Since the bounds provided by Kirrinnis are stated in a slightly different form (e.g.,~the polynomials have to rescaled) and since the paper seems to be not well known in the context of multipoint evaluation, we decided to give general bounds, which can directly be used in our analysis of the algorithm \cbisolve.}

\begin{theorem}
 \label{multipoint-evaluation}
  Let \(F \in \CC[x]\) be a polynomial of degree \(d\) with \(\norm{F}_1 \le 2^\mu\), where \(\mu \ge 1\),
  and let \(x_1, \dots, x_m \in \CC\) be points of absolute values bounded by \(2^\Gamma\), where \(m \le d\) and \(\Gamma \ge 1\). For an arbitrary positive integer \(L\), we can compute values \(\tilde{y}_i\in\CC\) such that \(\abs{\tilde{y}_i - F(x_i)} \le 2^{-L}\) for all \(i=1,\ldots,m\) using \(\sOh (d(L + \mu + d\,\Gamma))\) bit operations.
  The precision demand on the coefficients of \(F\) and the points \(x_i\) is bounded by \(L + \sOh(\mu + d\,\Gamma)\) bits after the binary point.
\end{theorem}
\begin{proof}
See \cref{appendix2}. For alternative approaches based on a division scheme from Sch\"{o}nhage~\cite{Schoenhage82}, see~\cite[Thm.~10]{KS13-multipoint} or~\cite[Thm.~3.9]{DBLP:journals/jc/Kirrinnis98}.
\end{proof}

\subsection{Review of the algorithm \cbisolve}
\label{cbisolve-review}

Our bivariate system solver \cbisolve extends the prior work \bisolve from~\cite{DBLP:conf/alenex/BerberichES11}. Most steps in \cbisolve are almost identical to the corresponding steps in \cbisolve, however, it computes isolating regions for all complex solutions, whereas \bisolve isolates only the real solutions. In addition, \cbisolve profits from the use of fast multipoint evaluation in the final validation phase which eventually yields a considerably improved overall complexity bound.

Because of the similarities to \bisolve, we give a self-contained but less detailed description of our algorithm in the paper at hand; for a deeper discussion of the key ideas of the algorithm and further details regarding an efficient implementation, the reader may consult the original paper.

We first give a brief outline of \cbisolve, where we emphasize on the differences when compared to its predecessor.
\cbisolve consists of two main stages:
\begin{itemize}
  \setlength{\itemsep}{0pt}
\item a \emph{projection phase}, where we project solutions onto their \(x\)- and \(y\)-coordinates using resultant computation and univariate root finding, and
\item a \emph{validation phase,} where we select the actual solutions of~\cref{eq:bisolve-system} among points in a candidate grid in \(\CC^2\) consisting of the preimages of the projected solutions.
\end{itemize}
We remark that \bisolve requires an additional intermediate \emph{separation phase} in which the output of the projection phase is refined to obey some separation criterion.
This is necessary because \bisolve uses a root isolation method that separates only the real roots of the corresponding resultant polynomials from each other, but the circumdisks of the isolating intervals are not allowed to contain nearby complex roots.
In contrast, there is no need for an additional separation step in \cbisolve as we use a complex root isolator in the projection phase and, thus, the corresponding requirements are automatically satisfied.
For achieving good complexity bounds, we propose to use the complex root isolator from~\cite{DBLP:conf/issac/MehlhornSW13} for the first step.
However, in practice, it can be replaced by any certified root solver, and we suggest to use numerical methods like~\cite{BF00,Kobel11} instead.
We now give details:

\paragraph{Projection phase}
\addcontentsline{toc}{subsubsection}{Projection phase}
The initial stage is to compute isolating disks in \(\CC\) for the \(x\)- and \(y\)-values of the solutions of~\cref{eq:bisolve-system}.
We define
\begin{align*}
  \begin{split}
    V^{(x)} &\coloneqq \set{\alpha\in\CC : \exists\; y\in\CC \cup \set{\infty} \text{ with } f(\alpha,y)=g(\alpha,y)=0} \qquad\text{and}\\
    V^{(y)} &\coloneqq \set{\beta \in\CC : \exists\; x\in\CC \cup \set{\infty} \text{ with } f(x,\beta)=g(x,\beta)=0},
  \end{split}
\end{align*}
where \(f(\alpha,\infty) = g(\alpha,\infty) = 0\) means that the leading coefficients \(f^{(y)}_{m_y}\) and \(g^{(y)}_{n_y} \in \CC[x]\) of \(f\) and \(g\), considered as polynomials in \((\CC[x])[y]\),
share the common factor \((x-\alpha)\).

The set \(V\) of solutions as defined in \cref{eq:solution-set} is a subset of the Cartesian product
\begin{align}
  \label{eq:candidate-set}
  \mathcal{C}:=V^{(x)} \times V^{(y)} \subset\CC^2,
\end{align}
which we denote the \emph{set of candidate solutions} for~\cref{eq:bisolve-system}.
\Itemref{res-prop:projection} states that \(V^{(x)}\) is exactly the set of roots of the resultant \(R^{(y)} \coloneqq \res(f,g;y)\) and, analogously, \(V^{(y)} = \set{\beta \in \CC : R^{(x)}(\beta) \coloneqq \res(f,g; x)(\beta) = 0}\).
We represent the zeros of \(R \in \set{ R^{(y)}, R^{(x)} }\) by a sets of \emph{isolating disks} \(D^{(x)}(\alpha) \coloneqq D_{r^{(x)} (\alpha)}(m^{(x)}(\alpha))\) and \(D^{(y)}(\beta) \coloneqq D_{r^{(y)} (\beta)}(m^{(y)}(\beta))\).
For the sake of readability, we omit the directions of the projections from now and write \(D(\alpha)\) for \(D^{(x)}(\alpha)\) or \(r(\beta)\) for \(r^{(y)}(\beta)\) and so on.

We say that the set of disks is \emph{isolating} if the (closed) disks \(D(\alpha)\) and \(D(\alpha')\) are disjoint for any two distinct roots \(\alpha\) and \(\alpha'\).
Note that, in general, \(\alpha \ne m(\alpha)\).
However, there is a one-to-one correspondence between the set of isolating disks \(D(\alpha)\) (and, thus, their centers \(m(\alpha)\)) and the roots \(\alpha\).
If the context is unambiguous, we exploit this correspondence and write \(\alpha\) instead of \(D(\alpha)\), hiding the fact that we only compute the isolating disks.

We additionally impose the following \emph{``well-iso\-la\-tion''} requirements on our isolating disks.
They ensure that some neighborhood of the boundaries of the disks \(D(\alpha)\) contains no root, and that the disks are not needlessly small.
The reason for these restrictions will become apparent in the next paragraph.
\begin{well-isol}
  \setlength{\itemsep}{0pt}
\item \label{well-isol:centered}
  Each \(\alpha\) is contained in the concentric disk \(D_{r(\alpha)/2}(m(\alpha))\) of halved radius;
\item \label{well-isol:margin}
  two distinct isolating disks \(D(\alpha)\) and \(D(\alpha')\) are separated by a margin of at least \(2\max\set{r(\alpha), r(\alpha')}\);
  and
\item \label{well-isol:precision}
  the size of each disk is related to the separation and the absolute value of the corresponding root. More precisely,
   \begin{align*}
    \tfrac{1}{32}\min\set{\separ_\alpha,M(\alpha)} \le r(\alpha) \le \tfrac{1}{4}\min\set{\separ_\alpha,M(\alpha)},
  \end{align*}
  where \(r(\alpha)\) is a power of two, and the precision of the center \(m(\alpha)\) matches the accuracy of the radius (up to an additive constant).
\end{well-isol}

Given such well-iso\-la\-ting disks, each solution of the initial system~\cref{eq:bisolve-system} is contained in a \emph{candidate polydisk} \(\Delta(\alpha,\beta) \coloneqq D(\alpha) \times D(\beta) \subset \CC^2\),
representing a candidate point \(\xi \coloneqq (\alpha,\beta) \in \mathcal{C}\) in~\cref{eq:candidate-set}, and each of these polydisks contains at most one solution (namely \(\xi\), if and only if \(\xi\) is a solution of~\cref{eq:bisolve-system}).

\newparagraph
Computing disks which obey the above restrictions from the output of any univariate root complex solver is straightforward: For a root \(\alpha\) of \(R\), let \(m_0(\alpha)\) and \(r_0(\alpha)\) denote the center and the radius of the isolating disks returned by the solver, respectively.
Assume that \(r_0(\alpha) \le \separ_\alpha/32\) for all \(\alpha\), a property that is typically ensured by design of the solver; otherwise, refine the disks further.
We remark that the solver presented in~\cite{DBLP:conf/issac/MehlhornSW13} returns such disks together with the corresponding multiplicities of the isolated roots by default.
For an arbitrary but fixed root \(\alpha\) of \(R\), let \(\alpha^\ast \coloneqq \arg\min_{\alpha'} \abs{m_0(\alpha) - m_0(\alpha')}\) denote a root of \(R\) whose isolating disk's midpoint is closest to \(m_0(\alpha)\).
Then,
\begin{align*}
  \abs{m_0(\alpha) - m_0(\alpha^\ast)} - r_0(\alpha) - r_0(\alpha^\ast) \le \separ_\alpha \le \abs{m_0(\alpha) - m_0(\alpha^\ast)} + r_0(\alpha) + r_0(\alpha^\ast)
\end{align*}
and, since both \(\separ_\alpha\) and \(\separ_{\alpha^\ast}\) are upper bounded by \(\abs{\alpha-\alpha^\ast}\),
\begin{align*}
  \separ_\alpha \in \big[\tfrac{29}{32}, \tfrac{35}{32}\big] \cdot \abs{m_0(\alpha)-m_0(\alpha^\ast)}.
\end{align*}
Hence,
\begin{align*}
  r(\alpha) \coloneqq \tfrac{1}{8} \cdot 2^{\lfloor \log\min \set{M(\alpha), \, \abs{m_0(\alpha)-m_0(\alpha^\ast)}} \rfloor}
\end{align*}
fulfills the inequality in \cref{well-isol:precision}.

We can now approximate the root \(\alpha\) by a complex number \(m(\alpha)\) with dyadic real and imaginary parts such that \(|m(\alpha)-\alpha|\le r(\alpha)\).
W.l.o.g., we can assume that \(m(\alpha)\) is chosen in a manner such that it has no more than \(1+\LOG (r(\alpha)^{-1})\) bits after the binary point.
Then, we define \(D(\alpha) \coloneqq D_{r(\alpha)}(m(\alpha))\), a disk that satisfies \crefrange{well-isol:centered}{well-isol:precision}.

\paragraph{Validation phase}
\addcontentsline{toc}{subsubsection}{Validation phase}
It remains to select the candidates that actually contribute to the solution set \(V\) and to discard the remaining ones.
As an \emph{exclusion predicate,} we simply use interval arithmetic with gradually increasing precision:
if \(\xi = (\alpha,\beta)\) is not a solution, sufficiently accurate refinement of the candidate polydisk and interval evaluation will eventually reveal that either \(0 \notin f(\Delta(\xi))\) or \(0 \notin g(\Delta(\xi))\).
However, proper solutions cannot be certified by interval arithmetic alone. We overcome this problem by a per-candidate sandwich bound argument on the resultant values. In what follows, we write \(n^\ast \coloneqq \max \set{m,n}\).

\newparagraph
Recall that, due to the identity from \itemref{res-prop:cofactor-rep}, we can represent the resultants as
\begin{align*}
  \begin{split}
    R^{(y)}(x) &= u^{(y)}(x,y) \cdot f(x,y) + v^{(y)}(x,y) \cdot g(x,y) \qquad\text{and}\\
    R^{(x)}(y) &= u^{(x)}(x,y) \cdot f(x,y) + v^{(x)}(x,y) \cdot g(x,y),
  \end{split}
\end{align*}
where the cofactor polynomials \(u^{(\ast)}\) and \(v^{(\ast)}\) are the determinants of Sylvester-like matrices \(U^{(\ast)}\) and \(V^{(\ast)}\).
Suppose that \(\LB(\alpha)\) and \(\LB(\beta) \) are positive lower bounds for the absolute values of \(R^{(y)}\) and \(R^{(x)}\), respectively, when restricted to the boundaries of \(D(\alpha)\) and \(D(\beta)\).
In addition, suppose that \(\UB_w(\xi)\) is an upper bound for the magnitude of \(w\) when restricted to \(\Delta(\xi)\), with \(w \in \set{\cramped{u^{(y)}, v^{(y)}, u^{(x)}, v^{(x)}}}\).
Then, a homotopy argument \cite[Thm.~4]{BEKS13} shows that at least one of the following inequalities holds for all \((x,y) \in \Delta(\xi)\) \emph{unless \(\xi\) solves the initial system~\cref{eq:bisolve-system}:}
\begin{align}
  \label{eq:sandwich-bound}
  \begin{split}
    \UB_{u^{(y)}}(\xi) \cdot \abs{f(x,y)} + \UB_{v^{(y)}}(\xi) \cdot \abs{g(x,y)} &\ge \LB(\alpha)
    \qquad\text{or}\\
    \UB_{u^{(x)}}(\xi) \cdot \abs{f(x,y)} + \UB_{v^{(x)}}(\xi) \cdot \abs{g(x,y)} &\ge \LB(\beta).
  \end{split}
\end{align}
Obviously, the converse holds as well: If \(\xi\) is a solution, then \(f\) and \(g\) converge to zero in a neighborhood of \(\xi\).
Hence, each point \(\xi_0 \in \Delta(\xi)\) that is sufficiently close to \(\xi\) violates~\cref{eq:sandwich-bound}, and thus, serves as a certificate for \(\xi\) being a solution.
In other words,~\cref{eq:sandwich-bound} serves as an \emph{inclusion predicate} for the candidates. It remains to show how to compute such lower and upper bounds \(\LB(\ast)\) and \(\UB_\ast\).

\newparagraph
The aforementioned ``well-iso\-la\-ting'' properties \cref{well-isol:centered,well-isol:margin} of an isolating disk guarantee that
\(\LB(\alpha) \coloneqq 2^{\lb(\alpha)}\) and \(\LB(\beta) \coloneqq 2^{\lb(\beta)}\), with
\begin{align}
  \label{def:LB}
  \begin{split}
    \lb(\alpha) &\coloneqq \left\lfloor\log\abs[\big]{R^{(y)}(m(\alpha)-r(\alpha))}\right\rfloor -\mult(\alpha,R^{(y)})-\deg R^{(y)} \qquad\text{and}\\
    \lb(\beta) &\coloneqq \left\lfloor\log\abs[\big]{R^{(x)}(m(\beta)-r(\beta))}\right\rfloor -\mult(\beta,R^{(x)})-\deg R^{(x)},
  \end{split}
\end{align}
define lower bounds on the magnitude of the resultants on the boundary of \(D(\alpha)\) and \(D(\beta)\).
For a simple proof of the latter claim, see \cref{lb-with-well-isol} in \cref{appendix1}.

For the upper bounds \(\UB_w\) for the cofactor polynomials, we define
\begin{align}
  \ub(\alpha) \coloneqq {}&\LOG \Big( (3 M(\alpha))^{n^\ast} \cdot \big( (m+1)^2 2^{\tau_f} (3M(\alpha))^m \big)^n \cdot \big( (n+1)^2 2^{\tau_g} (3M(\alpha))^n\big)^m \Big)\notag\\
  = {}&\sOh(mn \LOG \alpha + n\tau_f + m\tau_g)),\notag\\
  \ub(\beta) \coloneqq {}&\LOG \Big( (3 M(\beta))^{n^\ast} \cdot \big( (m+1)^2 2^{\tau_f} (3M(\beta))^m \big)^n \cdot \big( (n+1)^2 2^{\tau_g} (3M(\beta))^n\big)^m \Big)\notag\\
  = {}&\sOh(mn \LOG \beta + n\tau_f + m\tau_g)),\notag\\
 \UB(\alpha)\coloneqq {}& 2^{\ub(\alpha)},\qquad \UB(\beta)\coloneqq 2^{\ub(\beta)},\qquad\text{and}\notag\\
 \label{eq:def-UB}
 \UB(\xi) \coloneqq {} &\UB(\alpha) \cdot \UB(\beta) = 2^{\ub(\alpha)} \cdot 2^{\ub(\beta)}.
\end{align}
Using Laplace expansion along the last column and the row-wise Hadamard's determinant bound from \cref{goldstein-graham} on the remaining minors of \(U^{(\ast)}\) and \(V^{(\ast)}\),
it is straightforward to verify that the value \(UB(\xi)\) constitutes an upper bound for the absolute values of the cofactors \(u^{(\ast)}\) and \(v^{(\ast)}\), restricted to \(\Delta(\xi)\).

Notice that neither \(\LB(\alpha)\), \(\LB(\beta)\) nor \(UB(\xi)\) solely depend on the exact values of \(\alpha\) and \(\beta\) or the cofactor \(w\in \set{\cramped{u^{(y)}, v^{(y)}, u^{(x)}, v^{(x)}}}\), but rather on the inclusion disks as computed in the projection phase.

It should be remarked that, in practice, it is not advisable to compute the values \(\lb(\ast)\) and \(\ub(\ast)\) exactly.
Instead, it suffices to compute constant-factor approximations to the resultant values \(R^{(\ast)}(m(\ast)-r(\ast))\) and the magnitudes \(M(\alpha)\) and \(M(\beta)\)
and, hence, to conservatively under- and overestimate \(\lb(\ast)\) and \(\ub(\ast)\) by a small additive constant.
Our analysis in the following section will show that this will not harm the complexity results; however, we keep the above definitions for the sake of readability.

\newparagraph
In order to certify or discard the candidates, we could use a straightforward method based on interval arithmetic to certify or discard a candidate \(\xi=(\alpha,\beta)\). That is, we iteratively refine the corresponding inclusion polydisk \(\Delta(\xi)\) in stages with exponentially increasing precision (say, \(\rho = 1, 2, 4, 8, 16, \dots\) bits) and use interval arithmetic to evaluate \(f\) and \(g\) at \(\Delta(\xi)\).
In each iteration, we check whether interval arithmetic yields that \(0 \notin f(\Delta(\xi))\) or \(0 \notin g(\Delta(\xi))\) or, otherwise, whether there exists a point \(\xi_0\in\Delta(\xi)\) that violates~\cref{eq:sandwich-bound}.
Termination and correctness of the procedure follows from the above considerations.

The reader may notice that we could use fast approximate multipoint evaluation to group \(O(n)\) many of such evaluations together.
This is due to a natural correspondence between approximate evaluation and interval arithmetic:
if an input precision of \(\rho\) bits after the binary point allows us to approximate the value of the polynomial \(F\in\{f,g\}\) at the point \(\xi\) by some \(\tilde{F}\) with \(|F(\xi)-\tilde{F}|<2^{-L}\),
then \(|F(x)-\tilde{F}|<2^{-L}\) for all \(x\) with \(\norm{x-\xi}_\infty<2^{-\rho}\).
Thus, it is not surprising that fast approximate  multipoint evaluation techniques could save a factor \(n\) with respect to the computational complexity.

However, we can even do better by considering a more refined evaluation scheme:
Informally speaking, the worst-case complexity of a candidate is determined by the minimum of the ratios \(\LB(\alpha) / \UB(\alpha)\) and \(\LB(\beta) / \UB(\beta)\).
If, say, the former ratio is very small for a specific value \(\alpha^*\),
we have to expect a high precision demand until we can decide a candidate \((\alpha^*,\beta)\) irregardless of the fact that \(\beta\) might be ``nice''.
As a consequence, when carrying out all evaluations in a naive manner, we would have to approximate the coefficients of each polynomial \(f(\alpha,y)\) (and \(g(\alpha,y)\)) to the worst-case precision for each given fiber.
Then, \(O(\max\set{m,n})\) multipoint evaluations to the same precision are required to decide the \(O(mn)\) candidates over the fiber \(x = \alpha\)
in chunks of \(m\) or \(n\) candidates, respectively, for \(f(\alpha,y)\) and \(g(\alpha,y)\).

In contrast, we propose to first partition the candidates and then to apply \cref{multipoint-evaluation} when using multipoint evaluation.
We first define values
\begin{align}
  \label{eq:rho(alpha-beta)}
  \begin{split}
    \rho(\alpha) &\coloneqq \LOG (\UB(\alpha) / \LB(\alpha)) \qquad\text{and}\\
    \rho(\beta) &\coloneqq \LOG (\UB(\beta) / \LB(\beta)),
  \end{split}
\end{align}
which induce a partition of the set of candidates \((\alpha,\beta)\in\mathcal{C}\) into sets
\begin{alignat*}{2}
  \mathcal{C}^{(x)} &\coloneqq \set{ (\alpha,\beta) \in \mathcal{C} : \rho(\alpha) \ge \rho(\beta) } &\qquad&\text{and}\\
  \mathcal{C}^{(y)} &\coloneqq \set{ (\alpha,\beta) \in \mathcal{C} : \rho(\alpha) < \rho(\beta) } = \mathcal{C} \setminus \mathcal{C}^{(x)},
\end{alignat*}
Then, for each candidate in the first set,  the precision needed to process the candidate is directly related to the ``nicer'' ratio \(\LB(\alpha) / \UB(\alpha)\) corresponding to the \(x\)-coordinate of the candidate. Vice versa, for each candidate in the second set, the required precision is bounded in terms of the ratio \(\LB(\beta) / \UB(\beta)\) corresponding to the \(y\)-coordinate.

For an arbitrary but fixed \(\alpha\), the candidates \((\alpha, \beta) \in \mathcal{C}^{(x)}\) are now processed in rounds as follows: Initially, we define \(\rho:=1\) as the initial precision and
\(\mathcal{A}:=\{(\alpha,\beta)\in\mathcal{C}^{(x)}\}\) as the list of active candidates. In each round, we compute approximations of the polynomial
coefficients \(f^{(y)}_i(\alpha)\) and \(g^{(y)}_i(\alpha)\) of \(f(\alpha,y)\) and \(g(\alpha,y)\), respectively. Then, we use fast
approximate multipoint evaluation to approximate the values \(f(\alpha,\beta)\) and \(g(\alpha,\beta)\), where we consider blocks
of at most \(m\) and at most \(n\) distinct values \(\beta\) with \((\alpha, \beta) \in \mathcal{A}\). The precision
for the considered evaluations is chosen such that an absolute output precision of
\(\rho\)
bits after the binary point can be guaranteed.
We remark that there is no need to determine this precision a priori: it suffices to iteratively increase the working precision until the output precision of \(\rho\) is achieved.
If we double the precision in each step, the total cost is dominated (up to logarithmic factors) by the cost of the last run.

Candidates for which we can show that \(\abs{f(\alpha,\beta)} \ge 2^{-\rho}\) or \(\abs{g(\alpha,\beta)} \ge 2^{-\rho}\), implying \(f(\alpha,\beta)\neq 0\) or \(g(\alpha,\beta)\neq 0\), are discarded;
candidates that violate the inequalities in \cref{eq:sandwich-bound} are stored as solutions. Each candidate that is either discarded or stored
as a solution
is removed from \(\mathcal{A}\). We then continue with the next round, where we double the precision~\(\rho\). We stop if \(\mathcal{A}\) becomes empty. Following this approach, we can process all candidates \((\alpha, \beta) \in \mathcal{C}^{(x)}\). The candidate set \(\mathcal{C}^{(y)}\) is handled in almost the same manner, however, we aggregate evaluations along
horizontal fibers \(x=\beta\). That is, we first evaluate the coefficients \(f^{(x)}_i(y)\) and \(g^{(x)}_i(y)\) of \(f\) and \(g\)
(considered as polynomials in \(x\)) at the values \(y=\beta\), and then use multipoint evaluation along the horizontal fibers.

\subsection{Complexity analysis of \cbisolve}
\label{sec:bisolve-complexity}

We write \(N\) and \(T\), with \(N=m\cdot n\) and \(T=m\tau_g + n\tau_f + \Oh(m\log n + n\log m)\), for the explicit bounds on the degree and the bitsize of the resultants \(R^{(y)}\) and \(R^{(x)}\), as determined in \cref{resultant-magnitude}. In addition, we define \(n^\ast \coloneqq \max\set{m,n}\) and \(\tau^\ast \coloneqq \max\set{\tau_f, \tau_g}\).

\paragraph{Projection phase}
\addcontentsline{toc}{subsubsection}{Projection phase}
According to \cref{resultant-computation-complexity}, computing the resultants \(R \in \set{R^{(y)}, R^{(x)}}\) requires \(\sOh (n^\ast\mathstrut^3T)\) bit operations.
For isolating the roots of \(R\), we use the algorithm from~\cite[Thm.~5]{DBLP:conf/issac/MehlhornSW13}, which isolates the roots of \(R\) in time \(\sOh (N^3 + N^2T)\); cf.~\cref{univariate-root-isolation}.
Post-processing its output to satisfy the well-isolating properties \crefrange{well-isol:centered}{well-isol:precision} is possible within the same complexity bound; in fact, we can compute isolating disks of size \(2^{-L}\) in time \(\sOh(N^3+N^2T+NL)\).
From the well isolating properties, we further conclude that the midpoints and the radii of the disks \(D(\alpha)\) and \(D(\beta)\) are dyadic numbers representable by \(\sOh(N^2+NT)\) bits; cf.~\cref{bound-Sigma}.

\paragraph{Validation phase}
\addcontentsline{toc}{subsubsection}{Validation phase}
We now analyze the worst-case precision for which success of either the inclusion test or the exclusion test for a candidate \(\xi \coloneqq (\alpha,\beta)\) is guaranteed.
Define
\begin{align*}
  \delta(\xi) &\coloneqq \min\set{\LB(\alpha)/\UB(\alpha), \LB(\beta)/\UB(\beta)}.
\end{align*}
If \(\xi\) is a solution,~\cref{eq:sandwich-bound} is violated if and only if there exists a point \((x_0,y_0) \in \Delta(\xi)\) with
\begin{align}
  \label{eq:sandwich-bound-delta}
  \abs{f(x_0,y_0)} + \abs{g(x_0,y_0)} < \delta(\xi).
\end{align}
By contraposition,
\(\abs{f(x,y)} + \abs{g(x,y)} \ge \delta(\xi)\)
for all \((x,y) \in \Delta(\xi)\)
if \(\xi\) is not a solution.
Hence, in order to certify or discard the candidate \(\xi\), it suffices to approximate \(f(\xi)\) as well as \(g(\xi)\) to an error of less than \(\delta(\xi)/2\).
Namely, we can then either verify that at least one of the polynomials \(f\) and \(g\) does not vanish at \(\xi\), or verify that~\cref{eq:sandwich-bound-delta} holds.
From the definition \cref{eq:rho(alpha-beta)} of the values \(\rho(\alpha)\) and \(\rho(\beta)\) and the definition of the sets \(\mathcal{C}^{(x)}\) and \(\mathcal{C}^{(y)}\),
we conclude that each candidate \(\xi=(\alpha,\beta)\in \mathcal{C}^{(x)}\) (or \(\xi\in \mathcal{C}^{(y)}\)) can either be discarded or verified as a solution
if we approximate \(f(\xi)\) as well as \(g(\xi)\) to an error of less than \(2^{-\rho(\alpha)-1}\) (or \(2^{-\rho(\beta)-1}\)), since \(\rho(\alpha)\ge\rho(\beta)\) (or \(\rho(\alpha)<\rho(\beta)\)).

According to~\cref{multipoint-evaluation}, this implies that our inclusion/exclusion predicate succeeds for a candidate \(\xi=(\alpha,\beta)\in\mathcal{C}^{(x)}\) if we approximate the corresponding coefficient polynomials and run the multipoint evaluation with an absolute precision \(\rho(\xi)\) of size
\begin{align}
  \label{eq:rho(xi)}
  \rho(\xi) &= \rho(\alpha)+1 + \tilde{O}(\tau^\ast + n^\ast \LOG \alpha + n^\ast \LOG \beta) = \sOh (\rho(\alpha)),
\end{align}
where we used that the coefficients of \(f(\alpha,y)\) and \(g(\alpha,y)\) have a bitsize bounded by \(O(\tau^*+n^*\cdot\LOG\alpha)\).
To see the second bound, examine the definitions of \(\ub(\alpha)\) and \(\ub(\beta)\) and verify that \(n^\ast \LOG\beta \le \rho(\beta) \le \rho(\alpha)\) and \(n^\ast\LOG\alpha \le \rho(\alpha)\) for all \(\alpha\) and \(\beta\).

In order to carry out the above evaluations, it suffices to approximate each candidate \(\xi\in\mathcal{C}^{(x)}\) to at most
\(\tilde{O}(\max_\alpha \rho(\alpha))\)
bits after the binary point.
This can be achieved simultaneously for all candidates by approximating all roots of the resultant polynomials \(R^{(x)}\) and \(R^{(y)}\) to a corresponding number of bits after the binary point, a computation that uses
\(\cramped{\tilde{O}(N^3+N^2T + N \max_\alpha \rho(\alpha))}\)
bit operations according to \cref{univariate-root-isolation}, Part~2.

For an arbitrary but fixed \(\alpha\), the cost for computing sufficiently good approximations of the polynomials \(f(\alpha,y)\) and \(g(\alpha,y)\) is bounded by
\begin{align}
  \label{cost:coeffs}
  \tilde{O}\big(n^* (n^*\rho(\alpha)+\tau^*+n^*\log M(\alpha))\big)=\tilde{O}\big(n^{*2}\rho(\alpha)\big)
\end{align}
bit operations, and the same bound also holds for the cost of evaluating \(f(\alpha,y)\) and \(g(\alpha,y)\) at the \(O(mn)\) many points \(y=\beta\) for which \(\rho(\alpha)\ge \rho(\beta)\); cf.\ \cref{multipoint-evaluation}. Notice that, for approximating the polynomials \(f(\alpha,y)\) and \(g(\alpha,y)\), the factor \(n^*\) in \cref{cost:coeffs} is due to the fact that we have to evaluate up to \(n^*\) many coefficient polynomials. In contrast, for approximating the values \(f(\alpha,\beta)\) and \(g(\alpha,\beta)\), the factor \(n^*\) is related to the number of blocks (each consisting of at most \(n^*\) many points \(y=\beta\)) in the multipoint evaluation.
Now, summing up the bit complexity bound in \cref{cost:coeffs} over all \(\alpha\) and adding the cost for computing sufficiently good approximations of the candidates yields the following bound for the cost of processing all candidates \((\alpha,\beta)\in\mathcal{C}^{(x)}\):
\begin{align}
  \label{eq:time-fiber-polynomials}
   \sOh \Big(N^3+N^2T+N \cdot\max_\alpha \rho(\alpha) + n^\ast\mathstrut^2 \cdot\sum\nolimits_\alpha \rho(\alpha)\Big).
\end{align}

It remains to exploit amortization effects for the sum
\begin{align*}
  \sum\nolimits_\alpha \rho(\alpha)\le\sum\nolimits_{\alpha} \LOG \UB(\alpha)+\sum\nolimits_{\alpha}\LOG (\LB(\alpha)^{-1}).
\end{align*}
From the definition of \(\UB(\alpha)\), we conclude that each term in the first sum on the right hand side is upper bounded by \(\tilde{O}(N\LOG\alpha+T)\), and thus, the overall sum is upper bounded by \(\tilde{O}(NT)\) according to~\cref{bound-Mea} and the fact that there exist at most \(N\) many values \(\alpha\). For the sum related to the lower bounds, the considerations in \cite[Sec.~4.3.1]{DBLP:conf/issac/EmeliyanenkoS12} carry over%
\footnote{In \cite{DBLP:conf/issac/EmeliyanenkoS12}, the authors discuss the symmetric case where \(n = m\) and \(\tau = \tau_f = \tau_g\).
  However, their analysis only involves the magnitude of the resultants, hence the result stated here follows directly from substituting \(N\) and \(T\) for \(n^2\) and \(n\tau\), respectively.}
and yield
\begin{align*}
  \label{inequ:sumlowerbound}
  \begin{split}
    \sum\nolimits_\alpha \LOG (\LB(\alpha)^{-1}) &= \sOh \Big(\Sigma(R) + NT + \Mea(R) + N \sum\nolimits_\alpha \LOG \alpha\Big)\\
    &= \sOh(N^2 + NT);
  \end{split}
\end{align*}
for the sake of completeness, we provide a proof in \cref{proof:amortization-bound-LB} in \cref{appendix1}.
We conclude that \(\sum_\alpha \rho(\alpha)=\tilde{O}(N^2+NT)\).

Substituting the latter bound in \cref{eq:time-fiber-polynomials} shows that processing all candidates \((\alpha,\beta)\in\mathcal{C}^{(x)}\) needs \(\sOh \big(n^\ast\mathstrut^2 (N^2 + NT)\big)\) bit operations.
Finally, an entirely analogous argument as for \(\smash{\mathcal{C}^{(x)}}\) further shows that the latter bound also holds for the cost of processing the candidates in \(\mathcal{C}^{(y)}\). We summarize:

\begin{theorem}\label{thm:complexity-cbisolve}
  The algorithm \(\CC\)\textsc{BiSolve} computes isolating polydisks for all complex solutions of the system \(f(x,y) = g(x,y) = 0\),
  where \(f\), \(g \in \ZZ[x,y]\) are bivariate polynomials of magnitude \((m,\tau_f)\) and \((n,\tau_g)\) as defined as in~\cref{eq:bisolve-system}, in a total number of
  \begin{align*}
    \sOh \big({\max}^2\set{m,n} \cdot (m^2n^2 + mn (m\tau_g + n\tau_f))\big)
  \end{align*}
  bit operations.
  For refining all polydisks to a size of less than \(2^{-L}\), with \(L\) an arbitrary given positive integer, it needs
  \[
  \sOh \big({\max}^2\set{m,n} \cdot (m^2n^2 + mn (m\tau_g + n\tau_f))+mn\cdot L\big)
  \]
  bit operations.
  \end{theorem}

\subsection{Computation of a separating form}

In this section, we consider the problem of computing a \emph{separating form} for the polynomial system \cref{eq:bisolve-system}, that is, a
polynomial \(l_s(x,y)=x+s\cdot y\) such that \(\alpha+s\cdot \beta\neq \alpha'+s\cdot \beta'\) for each pair \((\alpha,\beta)\)
and \((\alpha',\beta')\) of distinct solutions of \cref{eq:bisolve-system}.
Typically, in most approaches for computing the solutions of a polynomial system, such an \(l_s\) is determined first,
followed by a shearing \((x,y)\mapsto (x+s\cdot y,y)\) to put the system into generic position (i.e.~no two solutions
share the same \(x\)-coordinate). In contrast, we assume that all solutions \(\xi_i\coloneqq (x_i,y_i)\in\CC\), \(i=1,\ldots,r\) with some \(r\le m\cdot n\), of the
input system are already computed using our algorithm \cbisolve and derive a separating form from sufficiently good approximations of the
solutions. This seems to be artificial at first glance since, usually, the main reason for computing a separating form is to solve the system. However,
in \cref{sec:topology} we will show that we can derive the topology of the algebraic curve defined as the zero set of \(f(x,y)\)
from the solutions of the systems \(f=f_y=0\) and \(f_x=f_y=0\) and a corresponding separating form for the latter system.

Since there exist at most
\(m\cdot n\) distinct solutions, there can be at most \(\binom{m\cdot n}{2}\) values for \(s\) yielding a linear form that is not separating. Namely, all such "bad" values for \(s\) must be among the set of all values
\begin{align}
s_{ij}\coloneqq \frac{x_i-x_j}{y_j-y_i}\in\CC,\quad\text{where }i,j\in\{1,\ldots,n\}\text{ and }y_i\neq y_j.
\end{align}
Hence, in order to compute a separating form, it suffices to approximate each value \(s_{ij}\) by some \(\tilde{s}_{ij}\in\CC\) with \(|s_{ij}-\tilde{s}_{ij}|<1/2\) and to choose an integer \(s\in\smash{\set[\big]{0,\ldots,\binom{m\cdot n}{2}}}\) which is not contained in any of the disks \(D_{1/2}(\tilde{s}_{ij})\). Namely, following this approach, we can exclude at most one integer value for \(s\) from each pair of solutions, and thus, there is at least one \(s\) left from the remaining integers that yields a separating form.

We estimate the cost for computing a separating form in the above described way: For approximating \(s_{ij}\) to an absolute error of \(1/2\), it suffices to
approximate \(x_i-x_j\) as well as \(y_i-y_j\) to a number of \(L_{ij}=O(\LOG \max(|x_i-x_j|,|y_i-y_j|)+\LOG ((y_i-y_j)^{-1}))\) bits after the binary point and to
divide the two values using approximate arithmetic with a precision of \(L_{ij}\). If we assume that corresponding approximations of the solutions \((x_i,y_i)\) and
\((x_j,y_j)\) are already given, then the cost for carrying out the two subtractions and the division is bounded by \(\tilde{O}(L_{ij})\).
Again, note there is no need to calculate the working precision a priori:
We can carry out the evaluations with an absolute precision \(L=1,2,4,8,\ldots\) until we can approximate \(\tilde{s}_{ij}\) to an absolute error of
\(1/2\), and we are guaranteed to succeed for some \(L\) of size \(O\big({\LOG \max(|x_i-x_j|,|y_i-y_j|)}+\LOG ((y_i-y_j)^{-1})\big)\).
The total cost is dominated by the last call.

In the previous section, we have shown that the pairwise distance between any two distinct values \(x_i\) and \(x_j\) as well as between any two distinct values
\(y_i\) and \(y_j\) is lower bounded by \(\smash{\cramped{2^{-\sOh(N^2+NT)}}}\), where we again write \(N = mn\) and \(T\), with \(T = \sOh(m\tau_g + n\tau_f)\), for the bounds on the degree and the bitsize of the resultants of \(f\) and \(g\).
Furthermore, the absolute value of all \(x_i\) and all \(y_i\) is upper
bounded by \(2^{\Oh(T)}\). Hence, it follows that \(L_{ij}=\tilde{O}(NT)\) for all pairs \((i,j)\), and thus, the solutions
\((x_i,y_i)\) have to be approximated to \(\tilde{O}(NT)\) bits after the binary points. The cost for computing such approximations of all
solutions is bounded by \(\sOh (n^\ast\mathstrut^2(N^2+NT))\) according to \cref{thm:complexity-cbisolve}, where \(n^\ast \coloneqq \max\set{m,n}\).
It remains to bound the cost for the
evaluations needed to compute the values \(\tilde{s}_{ij}\). For this, we fix some index \(i\) and sum up the precisions \(L_{ij}\) over all
\(j\) with \(y_j\neq y_i\):\pagebreak[1]
\begin{align*}
\sum_{j:y_j\neq y_i} L_{ij}&=O\Big(\sum_{\mathclap{j:y_j\neq y_i}}\LOG \max(|x_i-x_j|,|y_i-y_j|)+\sum_{\mathclap{j:y_j\neq y_i}}\LOG ((y_i-y_j)^{-1})\Big)\\
&=\sOh\Big(NT+\sum_{\hspace{-1em}\beta\neq y_i:R^{(x)}(\beta)=0\hspace{-1em}}\mult(\beta,R^{(x)})\cdot\LOG ((y_i-\beta)^{-1})\Big)\\
&=\sOh\Big(NT+\sum_{\hspace{-1em}\beta\neq y_i:R^{(x)}(\beta)=0\hspace{-1em}}\mult(\beta,R^{(x)})\cdot\LOG (\separ(\beta,R^{(x)})^{-1})\Big)\\
&=\sOh(N^2+NT),
\end{align*}
where \(R^{(x)}\coloneqq \operatorname{res}(f,g;x)\) and \(\beta\) runs over all distinct roots of \(R^{(x)}\) that are different from \(y_i\). For the second
inequality, we used that each \(y_j\) is a root of \(R^{(x)}\) and that there are at most \(\mult(\beta,R^{(x)})\) many pairs \((x_j,y_j)\) with \(\beta=y_j\); namely, each such solution contributes with at least one to the multiplicity of \(\beta\). For the last inequality, we used that, for an arbitrary integer polynomial
\(F\) of magnitude \((d,\mu)\), we have \(\smash{\sum_{z:F(z)=0}}\mult(z,F)\cdot\LOG(\separ(z,F)^{-1})=\tilde{O}(d^2+d\cdot\mu)\); see \cref{bound-Sigma}.
It follows that the sum over all \(L_{ij}\) with \(y_i\neq y_j\) is bounded by \(\sOh(N^3+N^2T)\) as there are at most \(N\)
solutions \((x_i,y_i)\). We summarize:

\begin{theorem}
A separating form \(l_s(x,y)=x+s\cdot y\) with \(s\in \set[\Big]{0,\ldots,\binom{m\cdot n}{2}}\) for the polynomial system \cref{eq:bisolve-system} can be
computed with a number of bit operations bounded by
\begin{align*}
  \sOh \big({\max}^2\set{m,n} \cdot (m^2n^2 + mn (m\tau_g + n\tau_f))\big).
\end{align*}
\end{theorem}
\subsection{Sign evaluation of a polynomial at the real solutions}

We next study the problem of evaluating the sign of a polynomial \(h\in\ZZ[x,y]\) at the \emph{real valued} solutions of the system \cref{eq:bisolve-system}.
In order to simplify the presentation, we assume throughout the following considerations that \(f\), \(g\), and \(h\) are integer polynomials of magnitude \((n,\tau)\).

We first consider the case, where \(h\) shares only a non-trivial factor with at least one of the polynomials \(f\) and \(g\).
Then, w.l.o.g., we can assume that \(\gcd(g,h)\in\ZZ\setminus\set{0}\); otherwise, switch \(f\) and \(g\).
According to \cref{thm:complexity-cbisolve}, we can use \cbisolve to compute isolating polydisks in \(\CC^2\) for all solutions of \(f=g=0\) as well as for \(g=h=0\) with \(\tilde{O}(n^6+n^5\tau)\) bit operations.
Recall that \cbisolve also computes the resultant polynomials \(R^{(y)}=\res(f,g;y)\in\ZZ[x]\) and \(R^{(x)}=\res(f,g;x)\in\ZZ[y]\) with \(R^{(y)}(\alpha)=R^{(x)}(\beta)=0\) and corresponding isolating (and
refineable) disks \(D(\alpha)\) and \(D(\beta)\) for the coordinates of each solution \((\alpha,\beta)\) of \(f=g=0\). For the system \(g=h=0\), it computes corresponding integer polynomials \(\bar{R}^{(y)}\coloneqq \res(g,h;y)\) and
\(\bar{R}^{(x)}\coloneqq \res(g,h;x)\) with \(\bar{R}^{(y)}(\bar{\alpha})=\bar{R}^{(x)}(\bar{\beta})=0\) and  corresponding isolating polydisks
\(D(\bar{\alpha})\) and \(D(\bar{\beta})\) for the coordinates of the solutions \((\bar{\alpha},\bar{\beta})\) of \(g=h=0\). Now, in order to determine the common solutions of \(f=g=0\) and
\(g=h=0\), we can simply compare the roots of the polynomials \(R^{(x)}\) and \(R^{(y)}\) with those of the polynomials \(\bar{R}^{(x)}\) and
\(\bar{R}^{(y)}\), respectively. In~\cite[Lem.~15]{DBLP:conf/issac/MehlhornSW13}, it has been shown that the cost for comparing the roots is
bounded by \(\tilde{O}(n^6+n^5\tau)\) bit
operations.\footnote{In~\cite[Lem.~15]{DBLP:conf/issac/MehlhornSW13}, it has been shown that we can
compare the roots of the resultant polynomials \(\res(f,f_y;y)\) and \(\res(f_x,f_y;y)\) with
\(\tilde{O}(n^6+n^5\tau)\) bit operations. However, the proof applies to arbitrary polynomials of comparable magnitude.}
It remains to evaluate the sign of \(h\) at those real valued solutions of \(f=g=0\) that are \emph{not}
solutions of \(g=h=0\). This can be achieved in a straightforward manner using approximate evaluation. More precisely, we approximate \(h(\alpha,\beta)\) to a precision of \(L=1,2,4,8,\ldots\) bits after the binary point. We stop increasing \(L\) as soon as we can decide the sign of \(h(\alpha,\beta)\). This is the case if the approximation of \(h(\alpha,\beta)\) has absolute value larger than \(2^{-L}\).
The cost for determining the sign of \(h(\alpha,\beta)\) is dominated by the call for the target precision and, thus, bounded by
\(\smash{\tilde{O}\big(n^2\big(\tau+\LOG (h(\alpha,\beta)^{-1})+n\cdot\LOG M(\alpha,\beta)\big)\big)}\) bit operations.

In order to bound the overall cost for the sign evaluations, we derive upper bounds for \(\sum_{(\alpha,\beta)}'\LOG(h(\alpha,\beta)^{-1})\) and for \(\sum_{(\alpha,\beta)}\LOG M(\alpha,\beta)\), where, in the first sum, we sum over the complex solutions \((\alpha,\beta)\) of \(f=g=0\) with \(h(\alpha,\beta)\neq 0\) and, in the second sum, over all complex solutions.
Since
\begin{align*}
  \sum_{(\alpha,\beta)}\LOG M(\alpha,\beta) & \le \sum_{\hspace{-1em}\alpha:R^{(y)}(\alpha)=0\hspace{-1em}}\mult(\alpha,R^{(y)})\cdot\LOG (\alpha)+\sum_{\hspace{-1em}\beta:R^{(x)}(\beta)=0\hspace{-1em}}\mult(\beta,R^{(x)})\cdot\LOG (\beta)\\
  &\le\LOG\Mea(R^{(y)})+\LOG\Mea(R^{(x)})+2n^2,
\end{align*}
it follows that \(\sum_{(\alpha,\beta)}\LOG M(\alpha,\beta)=\tilde{O}(n^2+n\tau)\).

For the bound on \(\sum_{(\alpha,\beta)}'\LOG (h(\alpha,\beta)^{-1})\), we can assume that the systems \(f=g=0\) and \(g=h=0\) are in generic position such that a solution \((\alpha,\beta)\) of \(f=g=0\) is also a solution of \(g=h=0\) if and only if \(\bar{R}^{(y)}(\alpha)=0\).
This can be achieved by considering a linear form \(x+s\cdot y\) (and a corresponding shearing \(x\mapsto x+s\cdot y\)) that is separating for the union of the solutions of \(f=g=0\) and \(g=h=0\).
Notice that we do not have to compute such a linear form;
we only need its existence with some \(s=O(\log n)\) for our argument to bound the sum \(\sum_{(\alpha,\beta)}'\LOG (h(\alpha,\beta)^{-1})\) which is invariant with respect to the coordinate transformation.

Now write \(R\) for \(R^{(y)}\) and \(\bar{R}\) for \(\bar{R}^{(y)}\), and let \(R^*\) and \(\bar{R}^*\) be the corresponding square-free parts.
The polynomial \(q(x)\coloneqq R^\ast / \gcd(R^\ast,\bar{R}^\ast)\) divides \(R^*\), and the roots of \(q\) are exactly the projections of all solutions \((\alpha,\beta)\) of \(f=g=0\) for which \(h\) does not vanish. Then, for each such \((\alpha,\beta)\), it holds that
\begin{align*}
0\neq \bar{R}(\alpha)=\bar{u}(\alpha,\beta)\cdot \smash{\underbrace{g(\alpha,\beta)}_{=0}}_{\rule{0em}{3.5ex}}+\bar{v}(\alpha,\beta)\cdot h(\alpha,\beta)=\bar{v}(\alpha,\beta)\cdot h(\alpha,\beta),
\end{align*}
where \(\bar{u}\) and \(\bar{v}\) are the cofactors in the cofactor representation of \(\bar{R}\); see also \cref{subsec:basics}. Hence, we have
\begin{align*}
  \LOG (h(\alpha,\beta)^{-1})&\le \LOG (\bar{v}(\alpha,\beta))+\LOG (\bar{R}(\alpha)^{-1})\\
  &\le \sOh(n\tau+n^2\LOG M(\alpha,\beta)+\LOG (\bar{R}(\alpha)^{-1})).
\end{align*}
When summing over all solutions \((\alpha,\beta)\) of \(f=g=0\), the first two terms sum up to \(\tilde{O}(n^4+n^3\tau)\) since there are at most \(n^2\) many solutions and the Mahler measure of \(R\) is bounded by \(2^{O(n^2+n\tau)}\). From the definition of \(q(x)\) and our genericity assumption, we conclude that the sum over the last term equals \(\smash{\sum_{\alpha:q(\alpha)=0} \LOG (\bar{R}(\alpha)^{-1})}\). Hence, since \(\abs{\bar{R}(\alpha)}\le \cramped{2^{O(n^2+n\tau+n^2\cdot \crampedLOG (\alpha))}}\) for all \(\alpha\), it follows that
\begin{align*}
  \sum_{\mathclap{\alpha:q(\alpha)=0}} \LOG (\bar{R}(\alpha)^{-1})&=\sum_{\mathclap{\alpha:q(\alpha)=0}} \log |\bar{R}(\alpha)|^{-1}+\sum_{\mathclap{\alpha:q(\alpha)=0}} O(n^2+n\tau+n^2\LOG \alpha)\\
  &=\log \Big(\prod_{\mathclap{\alpha:q(\alpha)=0}}|\bar{R}(\alpha)|^{-1}\Big)+O(n^4+n^3\tau)+n^2\cdot\LOG\Mea(q)\\
  &=\sOh(n^4+n^3\tau),
\end{align*}
where we used that \(q\), as a divisor of \(R\), has magnitude \((n^2,\sOh(n^2+n\tau))\) and that
\begin{align*}
  \prod_{\alpha:q(\alpha)=0}|\bar{R}(\alpha)^{-1}|= |\lcf(q(x))|^{\deg(\bar{R})}\cdot|\res(q,\bar{R})|^{-1}=2^{\sOh(n^4+n^3\tau)}.
\end{align*}
We conclude that the cost for evaluating the sign of \(h\) at all real valued solutions of \(f=g=0\) is bounded by \(\sOh(n^6+n^5\tau)\) many bit operations.

We are left to discuss the case, where \(h\) shares a non-trivial factor with both polynomials \(f\) and
\(g\). Suppose that \(p\coloneqq \gcd(g,h)\) is non-trivial and define \(h^*\coloneqq h/p\) and \(g^*\coloneqq g/p\).
Note that \(h^\ast\) and \(g^\ast\) are coprime, and since \(f\) and \(g\) are coprime, the same also holds for \(f\) and \(p\). According to~\cite[Lem.~13]{DBLP:conf/issac/MehlhornSW13},
we can compute \(p\), \(h^*\) and \(g^*\) with \(\sOh(n^6+n^5\tau)\) bit operations, and the magnitude of
these polynomials is bounded by \((n,O(n+\tau))\). The solutions of \(f=g=0\) now decompose into the
solutions of \(f=g^*=0\) and \(f=p=0\). Trivially, \(h=h^*\cdot p\) vanishes at all solutions of the latter system,
hence it remains to compute the sign of \(h\) at the solutions of \(f=g^*=0\). Based on the
considerations above, we can evaluate the sign of \(h^*\) as well as the sign of \(p\) at the solutions of
\(f=g^*=0\) with \(\sOh(n^6+n^5\tau)\) many bit operations as all involved polynomials have magnitude
\((n,O(n+\tau))\).

We summarize our results:
\begin{theorem}
Let \(f,g\in\ZZ[x,y]\) be coprime polynomials of magnitude \((n,\tau)\). Then, for an arbitrary polynomial \(h\in\ZZ[x,y]\) of magnitude \((n,\tau)\), we can evaluate the sign of \(h\) at all real-valued solutions of \(f=g=0\) with \(\sOh(n^6+n^5\tau)\) bit operations.
\end{theorem}
\section{Computing the topology of an algebraic plane curve}\label{sec:topology}

Based on our results on solving bivariate polynomial systems, we will show that, using a deterministic algorithm, we can compute the topology of a planar algebraic curve
\begin{align*}\mathcal{C}\coloneqq \{(x,y)\in\RR^2:f(x,y)=0\}\end{align*}
defined as the real zero set of a square-free integer polynomial \(f\in\ZZ[x,y]\) of magnitude
\((n,\tau)\) with \(\sOh(n^6+n^5\tau)\) bit operations. Our algorithm can be considered as a combination of the algorithm
\cbisolve and the randomized algorithm \textsc{TopNT} as introduced
in~\cite[Sec.~3]{BEKS13}. Hence, we only sketch our approach and refer to~\cite{BEKS13} for more details. Since the complexity analysis of our algorithm crucially depends on the results presented in~\cite{DBLP:journals/jsc/KerberS12,DBLP:conf/issac/MehlhornSW13}, we suggest to also consult these papers for amortized bounds on the complexity of the considered computations. As input, the algorithm receives the exact integer coefficients of the polynomial \(f\). Considering the following consecutive computations, it eventually returns a planar straight-line graph \(\mathcal{G}\) (embedded in \(\RR^2\)) that is isotopic to \(\mathcal{C}\):

\begin{enumerate}
\item Compute \(f_x^*\coloneqq f_x/\gcd(f_x,f_y)\) and \(f_y^*\coloneqq f_y/\gcd(f_x,f_y)\).
\item Determine an integer \(s\) of absolute value less than \(n^4\) such that
\begin{itemize}
\item[(a)] \(l_s=x+s\cdot y\) is a separating form for the \emph{strongly critical} points of \(f\), that is, the solutions of the system \(f_x^*=f_y^*=0\), and such that
\item[(b)] the leading coefficient of \(f(x+s y,y)\) with respect to \(y\) is a constant.
\end{itemize}
We use the term \emph{strongly critical} to denote the additional restriction \(f_x^*=f_y^*=0\) over \emph{critical} points, which satisfy \(f_x=f_y=0\).
\item Perform the coordinate transformation (\emph{shearing}) \(x\mapsto x+s\cdot y\), that is, replace the polynomial \(f(x,y)\) by \(F(x,y)\coloneqq f(x+sy,y)\). We define
\begin{align*}\bar{\mathcal{C}}\coloneqq \{(x,y)\in\RR^2:F(x,y)=0\}\end{align*}
to be the real vanishing set of the polynomial \(F(x,y)\).
Note that the shearing does not change the isotopy of the curve; hence, it will suffice to compute a straight-line graph isotopic to \(\bar{\mathcal{C}}\).
\item Isolate all real valued solutions \((\alpha_i,\beta_i)\), with \(i=1,\ldots,k\) and some
\(k\le n^2\), of the polynomial system \(F=F_y=0\). The points \((\alpha_i,\beta_i)\) and the (not necessarily distinct) values \(\alpha_i\) are called \emph{\(x\)-critical points} and \emph{\(x\)-critical values} of \(\bar{\mathcal{C}}\), respectively. W.l.o.g., we can assume that \(\alpha_1\le\cdots\le\alpha_k\).
\item Compute the sign of \(F_x\) at all \(x\)-critical points \((\alpha_i,\beta_i)\in\RR^2\) of
\(\bar{\mathcal{C}}\). Each \(x\)-critical point \((\alpha_i,\beta_i)\) with \(F_x(\alpha_i,\beta_i)=0\) is
a \emph{singular} point of \(\bar{\mathcal{C}}\).
\item Isolate all real roots \(\gamma_j\) of the polynomial \begin{align*}\hat{R}(x)\coloneqq {\frac{\partial R^*}{\partial x}} \;\bigg/\, {\gcd\bigg(\frac{\partial R^*}{\partial x},\frac{\partial^2 R^*}{\partial x^2}\bigg)}\end{align*}
which is the square-free part of the derivative of \(R^*\), and \(R^*\coloneqq R\,/\gcd\big(R,\frac{\partial{R}}{\partial x}\big)\) is defined as the square-free part of the resultant \(R\coloneqq \res(F,F_y;y)\).

The values \(\gamma_1,\ldots,\gamma_m\) of \(\hat{R}\) separate the roots of \(R\), that is, in between of two
distinct consecutive roots \(\alpha_i\) and \(\alpha_{i+1}\) of \(R\), there exists at least one root \(\gamma_j\) of
\(\hat{R}\). If necessary, we discard some \(\gamma_j\) in an arbitrary manner such that, in between of two consecutive distinct
values \(\alpha_i\) and \(\alpha_{i+1}\), there exists exactly one \(\gamma_i\).
(Notice that it is not guaranteed that the \(\bar{\mathcal{C}}\) is in generic position with respect to its \(x\)-critical values, that is, there might be two distinct \(x\)-critical points \((\alpha_i,\beta_i)\) and \((\alpha_{i+1},\beta_{i+1})\) with \(\alpha_i=\alpha_{i+1}\).)

We remark that, for a practical implementation of the algorithm, we propose to consider arbitrary rational
values \(\gamma_i\) in between each pair of consecutive real roots of \(R\). However, in general, this yields a worse
complexity bound for the root isolation in the next step.
\item Isolate the real roots of all polynomials \(F(\alpha_i,y)\) and \(F(\gamma_i,y)\).
\item Connect points \((\alpha_i,y^*)\in\bar{\mathcal{C}}\) (or \((\alpha_{i+1},y^*)\in\bar{\mathcal{C}}\)) and \((\gamma_i,y^{**})\in\bar{\mathcal{C}}\) by a line segment if and only if they are connected
via an arc of \(\bar{\mathcal{C}}\). Return the so obtained planar straight-line graph \(\mathcal{G}\), which
is isotopic to \(\bar{\mathcal{C}}\), and thus also isotopic to \(\mathcal{C}\).
\end{enumerate}

In the previous sections, we have shown how to perform the first five steps (except for Step 2 (b))
with \(\tilde{O}(n^6+n^5\tau)\) bit operations, considering that the polynomials \(f\),
\(f_x^*\) and \(f_y^*\) as well as the transformed polynomials \(F(x,y)=f(x+sy,y)\), \(F_x(x,y)=f_x(x+sy,y)\) and \(F_y(x,y)=s\cdot f_x(x+sy,y)+f_y(x,y)\) have magnitude \((n,O(n \log n +\tau))\),
because the bitsize of \(s\) is bounded by \(O(\log n)\).
For the computation needed to guarantee 2 (b), we remark that all except \(n\) bad values for \(s\) fulfill 2 (b) and that we can compute these values with a number of bit operations bounded by \(\sOh(n^3+n^2\tau)\). Namely, these bad values are exactly the roots of the leading coefficient \(f_n^{(y)}(s)\in\ZZ[s]\) of \(f(x+sy,y)\in\ZZ[s,x][y]\) with respect to \(y\). Hence, the computation of \(f_n^{(y)}(s)\) as well as approximating all of its roots to an error of less than \(1/2\) needs no more than \(\sOh(n^3+n^2\tau)\) bit operations.  For Step 6, we refer to~\cite[Sec.~3]{DBLP:conf/issac/MehlhornSW13}, where a bit complexity bound of \(\sOh(n^6+n^5\tau)\) has been given.

The last two steps of the above algorithm need more explanation: According to~\cite[Sec.~3.2.2]{BEKS13}, there exists
no solution of \(f=\gcd(f_x,f_y)=0\) in \(\CC^2\). That is, there are either no solutions (iff \(\gcd(f_x,f_y)\) is trivial) or all solutions are located at infinity. Hence, all singular points of
\(\mathcal{C}\) are the common solutions of \(f=f_x^*=f_y^*=0\). From the
transformation in Step 3, we conclude that there are no two complex solutions of \(F_x^*=F_y^*=0\)
(i.e.~strongly critical points of \(F\)) sharing the same \(x\)-coordinate, where
\(F_x^*\coloneqq F_x/\gcd(F_x,F_y)\) and \(F_y^*\coloneqq F_y/\gcd(F_x,F_y)\). Namely, the \emph{strongly} critical points of \(F\) are
directly obtained from shearing the strongly critical points of \(f\).
Notice that this does not
hold for the \(x\)-critical points of \(f\) and \(F\); in particular,
choosing a separating form \(x+sy\) for the system \(f=f_y=0\) does not imply that the sheared curve
\(\bar{\mathcal{C}}\) is in generic position with respect to its \(x\)-critical points. However, if
we choose \(s\) according to the requirements in Step 2,
then \(\bar{\mathcal{C}}\) is in generic position with respect to its \emph{strongly} critical points.
In particular, this means that for each \(x\)-critical value \(\alpha_i\) of \(\bar{\mathcal{C}}\), there exists at most
one strongly critical point above \(\alpha_i\). Furthermore, from the computation in Step 5, we
can also determine whether there exists a singular point with \(x\)-coordinate \(\alpha_i\) or not.
At this point, we remark that each singular point above \(\alpha_i\) must be real valued because
there exists at most one strongly critical point above \(\alpha_i\) and \(F\) has real
valued coefficients, which implies that the critical points above \(\alpha\) must arise in complex conjugate pairs.
We can then use the computation in~\cite[Sec.~3.2.2.]{BEKS13} to determine the number \(n_{\alpha_i}\) of distinct
complex roots of \(F(\alpha_i,y)\). Namely, it holds that
\begin{align*}
  &\mathllap{n_{\alpha_i}}=\begin{cases}
    n-\mult(\alpha_i,R)+\mult(\alpha_i,Q),  & \text{if there exists a singular point }(\alpha_i,\beta)\in\bar{\mathcal{C}}\\
    n-\mult(\alpha_i,R), & \text{if there exists no singular point }(\alpha_i,\beta)\in\bar{\mathcal{C},}
  \end{cases}
\end{align*}
where \(R\coloneqq \res(F,F_y;y)\) and \(Q\coloneqq \res(F_x^*,F_y^*;y)\). In other words, \(n_{\alpha_i}\) can be derived from the multiplicity of \(\alpha_i\) as a root of \(R\) (and \(Q\)) and
the fact whether there exists a singular point above \(\alpha_i\) or not.
For computing the values \(\mult(\alpha_i,R)\) and \(\mult(\alpha_i,Q)\) for all \(\alpha_i\), we refer
to~\cite[Lem.~15]{DBLP:conf/issac/MehlhornSW13}, where it has been shown that the
cost for the necessary computations is bounded by \(\sOh(n^6+n^5\tau)\) bit operations. Since each polynomial \(F(\gamma_i,y)\) has degree \(n\) and all its roots are simple, the number of distinct roots of \(F(\gamma_i,y)\) equals \(n\) for all \(i\).
Now, having computed the number of distinct roots for each of the polynomials \(F(\alpha_i,y)\) and \(F(\gamma_i,y)\), we can use the algorithm from~\cite{DBLP:conf/issac/MehlhornSW13} to compute isolating intervals for all real roots of the latter polynomials together with the corresponding
multiplicities. Using \cref{univariate-root-isolation}, one can show that the cost for the root isolation and for computing sufficiently good approximations of the roots of \(R\) and \(\hat{R}\) is bounded by \(\sOh(n^6+n^5\tau)\) bit operations; for more details, we refer to~\cite[Lem.~20]{DBLP:conf/issac/MehlhornSW13}. After Step 7,
we have already computed the vertices of the planar graph \(\mathcal{G}\), and each of these vertices is located on the curve \(\bar{\mathcal{C}}\). It remains to show how to
connect the vertices by line segments in an appropriate manner. That is, we have to determine whether two vertices in neighboring fibers are connected via an arc of \(\bar{\mathcal{C}}\) and then connect them by a line segment if and only if the latter is the case.
Notice that, for each value \(x=\gamma_i\) that is not \(x\)-critical, we have computed isolating intervals for
all real roots of the polynomial \(F(\gamma_i,y)\), and each such root is simple. For each \(x\)-critical value \(x=\alpha_i\), we have also computed isolating intervals for all real roots \(y_{i,1},\ldots,y_{i,m_i}\) of
\(F(\alpha_i,y)\), with \(m_i\le n_{\alpha_i}\), and, in addition, we know
\begin{itemize}
\item the multiplicity \(\mu_{i,j}:=\mult(y_{i,j},F(\alpha_i,y))\) of each root \(y_{i,j}\),
\item the sign of \(F_x(\alpha_i,y_{i,j})\) for each \(y_{i,j}\) with \(\mu_{i,j}>1\), and, in particular,
\item whether \((\alpha_i,y_{i,j})\) is a singular point or not.
\end{itemize}
For each \(x\)-critical point \((\alpha_i,y_{i,j})\) that is not singular, we can further compute the
sign of the \(\mu_{i,j}\)-th partial derivative \(\cramped{F_y^{(\mu_{i,j})}}\) at \((\alpha_i,y_{i,j})\) with respect to \(y\) by approximating \((\alpha_i,y_{i,j})\) and approximately evaluating \(F_y^{(\mu_{i,j})}(\alpha_i,y_{i,j})\), where we define \(F_y^{(k)} \coloneqq \frac{\partial^k}{\partial y^k} F\).
The analysis in~\cite{DBLP:conf/issac/MehlhornSW13} shows that the necessary computations can be carried out using \(\tilde{O}(n^6+n^5\tau)
\) bit operations. More precisely, the cost for evaluating \(F\) at some point \((x_0,y_0)\) with an
output precision of \(L\) bits after the binary point is bounded by \(\tilde{O}(n^2(L+n+\tau+n\LOG \max(\abs{x_0},\abs{y_0})))\) and the sum of all values \(\LOG(F_y^{(\mu_{i,j})}(\alpha_i,y_{i,j})^{-1})\)
and \(\LOG \max(\abs{\alpha_i},\abs{y_{i,j}})\) is bounded by \(\tilde{O}(n^4+n^3\tau)\); see~\cite[Lem.~17 and~18]{DBLP:conf/issac/MehlhornSW13}. Also, the cost for computing good enough approximations of the roots \(y_{i,j}\)
is bounded by \(\tilde{O}(n^6+n^5\tau)\); see the proof of~\cite[Lem.~20]{DBLP:conf/issac/MehlhornSW13} for details.

We claim that the above information is already sufficient to determine (in a purely combinatorial
way) whether two points in the fibers \(x=\alpha_i\) and \(x=\gamma_i\) are connected via an arc of \(
\bar{\mathcal{C}}\): Namely, the local topology of \(\bar{\mathcal{C}}\) at each point that is not \(x\)-critical is trivial, that is, there exists exactly one arc that enters the point from the left and
leaves to the right. For each \(x\)-critical point \(p=(\alpha_i,\beta_i)\) that is not singular, there is exactly one arc incident to \(p\) as well, and one of the following three possibilities for the local topology at \(p\);
see also \cref{fig:connection} for an illustration:
\setlength\columnsep{2em}
\begin{wrapfigure}[23]{r}{\widthof{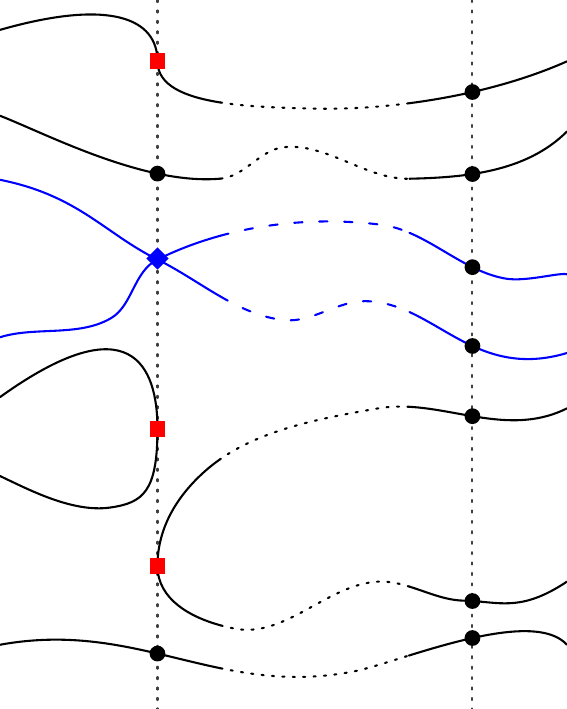}}
  \centering\input{connection.pdf_tex}
  \caption{Connecting points in a critical fiber \(x=\alpha_i\) (on the left) with points in a neighboring non-critical fiber \(x=\gamma_i\) (on the right) is purely combinatorial
    if there exists at most one singular point (blue diamond) and the local topology at the non-singular \(x\)-critical points (red squares) is known.}
  \label{fig:connection}
\end{wrapfigure}
\vspace{-1em}
\begin{point-types}
\item The arc enters \(p\) from the left and leaves to the right,
\item the arc enters \(p\) from the left and leaves to the left, or
\item the arc enters \(p\) from the right and leaves to the right.
\end{point-types}
In the latter two cases, the point \(p\) is an \(x\)-extremal point, that is, the curve \(
\bar{\mathcal{C}}\) makes a turn at \(p\). In the first case, the local topology of \(\bar{\mathcal{C}}\)
is the same as at any point that is not \(x\)-critical, however, the arc passing through \(p\) is vertical at \(p\).
Obviously, the first case applies if and only if the
multiplicity \(\mu=\mult(\beta_i,F(\alpha_i,y))\) of \(\beta_i\) as a root of \(F(\alpha_i,y)\) is odd. If
\(\mu\) is even, then case~2 applies if and only if \(F_x(\alpha_i,\beta_i)\cdot F_y^{(\mu)}
(\alpha_i,\beta_i)>0\), and case~3 applies if and only if \(F_x(\alpha_i,\beta_i)\cdot F_y^{(\mu)}
(\alpha_i,\beta_i)<0\).\pagebreak[1]

In summary, for each point \(p=(\alpha_i,\beta_i)\) in the fiber \(x=\alpha_i\) except for a unique
singular point (if such a point exists), we know the local topology of \(\bar{\mathcal{C}}\) at \(p\).
Hence, we can determine all connections between the non-singular points in an \(x\)-critical fiber \(x=\alpha_i\) and
the points in a neighboring fiber \(x=\gamma_i\) (that is not \(x\)-critical) from bottom to top and
from top to bottom until only the connections to the singular point are left. Then, all points in
the fiber \(x=\gamma_i\) that are not connected yet with a point in the fiber \(x=\alpha_i\) must be
connected with the unique singular point.

We conclude that the topology of \(\mathcal{C}\) can be computed with a number of bit operations bounded by \(\tilde{O}(n^6+n^5\tau)\). Since we can always compute the square-free part \(f^*\) of an arbitrary bivariate polynomial \(f\in\ZZ[x,y]\) of magnitude \((n,\tau)\) with \(\tilde{O}(n^6+n^5\tau)\) many bit operations without affecting the zero set, and since \(f^*\) has magnitude \((n,O(n\log n+\tau))\), we obtain the following general result:

\begin{theorem}
Given an arbitrary (not necessarily square-free) polynomial \(f\in\ZZ[x,y]\) of magnitude \((n,\tau)\), we can compute the topology
of the real planar algebraic curve
\begin{align*}
  \mathcal{C} \coloneqq \set{(x,y)\in\RR^2:f(x,y)=0}
\end{align*}
with a number of bit operations bounded by \(\tilde{O}(n^6+n^5\tau)\).
\end{theorem}
\appendix
\section{Missing Proofs}\label{appendix1}

\begingroup
\renewcommand\thetheorem{\ref{bound-Sigma}\relax}
\makeatletter
\let\c@theorem\c@thmt@dummyctr
\let\theHtheorem\theHthmt@dummyctr
\let\label=\@gobble
\let\ltx@label=\@gobble
\makeatother
\begin{theorem}
  Let \(F\in\ZZ[x]\) be an integer polynomial of magnitude \((d,\mu)\) with distinct roots \(z_1,\ldots,z_m\in\CC\). Then, it holds that
  \begin{align*}
    \sum\nolimits_{i=1}^m \mult_i\LOG\separ(z_i,F)^{-1} = \sOh(d^2+d\mu).
  \end{align*}
  In particular,
  \begin{align*}
    \Sigma^\ast(F)
    &\le \Sigma(F)
    = \sOh(d^2 + d\mu).
  \end{align*}
\end{theorem}
\endgroup

\begin{proof}
We consider the factorization of \(F\) (over \(\mathbb{Z}\)) into square-free and pair-wise coprime factors:
\begin{align*}
  F(x)=\prod\nolimits_{i=1}^k Q_i(x)^{s_i},\ \ \text{with } d_i:=\deg(Q_i)\ge 1,
\end{align*}
such that the polynomials \(Q_i(x)\) and \(F(x)/Q_i(x)^{s_i}\) are coprime, and \(d=\sum_{i=1}^k d_is_i\).
We further denote \(F^*\) the square-free part of \(F\) and
\(m:=\deg(F^*)=\sum_{i=1}^k d_i\) its degree. Then, for arbitrary roots \(\alpha\)
and \(\beta\) of \(F^*\), it holds that
\begin{align*}
\abs{(F^*)'(\alpha)} & = \abs{\lcf(F^*)}\cdot\abs{\alpha-\beta}\cdot\prod_{\hspace{-1em}\gamma\ne\alpha,\,\beta:F^*(\gamma)=0\hspace{-1em}}\abs{\gamma-\alpha}\\
& \le\abs{\lcf(F^*)}\cdot \abs{\alpha-\beta}\cdot\prod_{\hspace{-1em}\gamma\ne\alpha,\,\beta:F^*(\gamma)=0\hspace{-1em}}2 M(\alpha,\gamma)\\
&\le 2^{m-2}\cdot\abs{\alpha-\beta}\cdot M(\alpha)^{m-3}\cdot\Mea(F^*)
\end{align*}
since \(\Mea(F^*)=\abs{\lcf(F^*)}\cdot\prod_{z:F^*(z)=0} M(z)\).
Suppose, w.l.o.g., that \(\alpha\) is a root of \(Q_i\) and
\(\beta\) is a root of \(F^*\) closest to \(\alpha\). Then, according to the above inequality, we have
\begin{align*}
  \separ(\alpha,F)=\abs{\alpha-\beta}\ge \frac{\abs{(F^*)'(\alpha)}}{2^{m-2}\cdot M(\alpha)^{m-3}\cdot\Mea(F^*)}.
\end{align*}
We now apply this inequality to the product over all \(\separ(\alpha_j,F)\), \(j=1,\ldots,d_i\), where \(\alpha_1,\ldots,\alpha_{d_i}\) denote the roots of \(Q_i\):
\begin{align}
  \label{inequality1}
  \begin{split}
    \prod_{j=1}^{d_i} \separ(\alpha_j,F) &\ge 2^{(2-m)d_i}\cdot\Mea(Q_i)^{3-m}\cdot\Mea(F^*)^{-d_i}\cdot\prod_{j=1}^{d_i}\abs{(F^*)'(\alpha_j)}\\
    &=2^{(2-m)d_i}\cdot\Mea(Q_i)^{3-m}\cdot\Mea(F^*)^{-d_i}\cdot\prod_{j=1}^{d_i}\abs[\Big]{(Q_i)'(\alpha_j)\cdot\frac{F^*}{Q_i}(\alpha_j)}
  \end{split}
\end{align}
since
\begin{align*}
  (F^*)'(\alpha_j)&={\underbrace{Q_i(\alpha_j)}_{=0}\cdot \left(\frac{F^*}{Q_i}\right)'(\alpha_j)+(Q_i)'(\alpha_j)\cdot\frac{F^*}{Q_i}(\alpha_j).}
\end{align*}
In addition, we have
\begin{align*}
  \prod_{j=1}^{d_i}\abs{Q_i'(\alpha_j)} &= \abs{\lcf(Q_i)^{1-d_i}\cdot\operatorname{res}(Q_i,Q_i')}\ge \abs{\lcf(Q_i)}^{1-d_i} \qquad\text{and}\\
  \prod_{j=1}^{d_i}\abs[\Big]{\frac{F^*}{Q_i}(\alpha_j)}&=\abs[\Big]{\lcf(Q_i)^{d_i-m}\cdot\res\Big(Q_i,\frac{F^*}{Q_i}\Big)}\ge \abs{\lcf(Q_i)}^{d_i-m},
\end{align*}
since \(\res(Q_i,Q_i')\) and \(\res(Q_i,\frac{F^*}{Q_i})\) are non-zero integers.
Applying the latter two inequalities to \cref{inequality1} now yields
\begin{align*}
\prod_{j=1}^{d_i} \separ(\alpha_j,F)\ge 2^{(2-m)d_i}\cdot\Mea(Q_i)^{3-d}\cdot
\Mea(F^*)^{-d_i}\cdot\abs{\lcf(Q_i)^{1-m}}.
\end{align*}
Finally, consider the product of the separations of all roots to the respective powers \(s_i\):
\begin{align*}
\prod_{i=1}^k\prod_{j=1}^{d_i}\separ(\alpha_j,F)^{s_i}&\ge
\prod_{i=1}^k 2^{(1-m)d_is_i}\cdot\Mea(Q_i)^{(3-m)s_i}\cdot\Mea(F^*)^{-d_is_i}\cdot\prod_{i=1}^k\abs{\lcf(Q_i)}^{-s_i}\\
&=2^{(1-m)d}\cdot\Mea(F)^{3-m}\cdot\Mea(F^*)^{-d}\cdot\abs{\lcf(F)}^{-1}= 2^{-\tilde{O}(d^2+d\mu)},
\end{align*}
where we used that \(\prod_{i=1}^k \Mea(Q_i)^{s_i} = \Mea(F)\) by the multiplicativity of the Mahler measure and \(\Mea(F^*)\le\Mea(F)= 2^{O(\mu+\log d)}\).
This shows that \(\sum_{j=1}^m \log\separ(z_i,F)^{-1}=\sOh(d^2+d\mu)\). Since, for each root \(z_j\) of \(F\), \(\separ(z_j,F)\) is upper bounded by two times the maximal absolute value of the roots of \(F\), we have \(\separ(z_j,F)<2^{\mu+2}\) according to the Cauchy root bound. Thus, it follows that \(\sum_{j=1}^m \LOG\separ(z_j,F)^{-1}\le d(\mu+2)+\sum_{j=1}^m \log\separ(z_j,F)^{-1}=\sOh(d^2+d\mu)\).
\end{proof}

\begin{lemma}
  \label{lb-with-well-isol}
Suppose that the properties \cref{well-isol:centered,well-isol:margin} are fulfilled. Then, the values \(\LB(\alpha)\) and \(\LB(\beta)\) as defined in \cref{def:LB} constitute lower bounds for the absolute values of the resultants \(R^{(y)}\) and \(R^{(x)}\) restricted to the boundary of \(D(\alpha)\) and \(D(\beta)\), respectively.
\end{lemma}
\begin{proof} The proof is almost identical to the proof of~\cite[Lem.~3.1]{DBLP:conf/alenex/BerberichES11}.
  Write \(m \coloneqq m(\alpha)\), \(r \coloneqq r(\alpha)\), and \(R \coloneqq R^{(y)}\).
  According to \cref{well-isol:centered}, \(D_{r/2}(m)\) is also isolating, hence
  \(\frac{\abs{z-\alpha}}{\abs{(m-r)-\alpha}} > \frac{r/2}{3r/2} > \frac{1}{4}\)
  for all points \(z \in \partial D(\alpha)\) on the boundary of \(D(\alpha)\).
  In addition, \cref{well-isol:margin} guarantees that \(R(m-r) \ne 0\) and that, for any root \(\alpha' \ne \alpha\) of \(R\), it holds that
  \(\frac{\abs{z-\alpha'}}{\abs{(m-r)-\alpha'}} \ge \frac{\abs{\alpha'-m}-r}{\abs{\alpha'-m}+r} \ge 1 - \frac{2r}{4r} = \frac{1}{2}\).
  Hence, it follows that
  \begin{align*}
    \frac{\abs{R(z)}}{\abs{R(m-r)}}
    &= \Big(\frac{\abs{z-\alpha}}{\abs{(m-r)-\alpha}}\Big)^{\mult_\alpha} \, \cdot \! \prod_{\hspace{-1.2em}\zeta \ne \alpha : R(\alpha') = 0\hspace{-1.2em}}\quad \Big(\frac{\abs{z-\alpha'}}{\abs{(m-r)-\alpha'}}\Big)^{\mult_{\alpha}'}\\
    &> 4^{-\mult_\alpha} 2^{-(\deg R -\mult_\alpha)} = 2^{-\mult_\alpha - \deg R}.
    \tag*\qedhere
  \end{align*}
\end{proof}

\clearpage
\begin{lemma}\label{proof:amortization-bound-LB}
Let \(\LB(\alpha)\) be defined as in \cref{def:LB}.
Then,
\begin{align*}
\sum\nolimits_\alpha \LOG (\LB(\alpha)^{-1}) =
\sOh(N^2 + NT),
\end{align*}
where we sum over all distinct complex roots of the resultant polynomial \(R^{(y)}\).
\end{lemma}
\begin{proof}
  For any root \(\alpha\) of the resultant \(R \coloneqq R^{(y)}\), write \(m \coloneqq m(\alpha)\), \(r \coloneqq r(\alpha)\), and
  \begin{align*}
    R(x) = \lcf R \cdot (x-\alpha)^{\mult_\alpha} \cdot \smash{\prod\nolimits_{\zeta \ne \alpha}} (x-\zeta)^{\mult_\zeta},
  \end{align*}
  where \(\zeta\) runs over all distinct complex roots of \(R\) but \(\alpha\).
  Define \(\sigma_\alpha \coloneqq \min\set{\separ_\alpha, M(\alpha)}\).
  Property \cref{well-isol:precision} ensures that \(\sigma_\alpha/16 \le r \le \sigma_\alpha / 4\). Hence, it follows that
  \begin{alignat*}{2}
    \tfrac{1}{32}\sigma_\alpha \le \tfrac{1}{2}r &\le \abs{(m-r)-\alpha} &&\le \tfrac{3}{2}r \le \tfrac{3}{8} \sigma_\alpha < \tfrac{1}{2}\sigma_\alpha
    \qquad\text{and}\\[1ex]
    \tfrac{1}{2} \abs{\alpha-\zeta} < \abs{\alpha-\zeta}-\tfrac{1}{2} r
    &\le \abs{(m-r)-\zeta} &&\le \abs{\alpha-\zeta}+\tfrac{3}{2}r \le 2\abs{\alpha-\zeta}
  \end{alignat*}
  for all roots \(\zeta \ne \alpha\) of \(R\).
  Hence,
  \begin{align}
    \LB(\alpha) &\le 2^{-\mult_\alpha - \deg R} \cdot \abs{R(m-r)} \notag\\
    &< 2^{-N} \cdot \abs{\lcf R} \cdot \abs{(m-r)-\alpha}^{\mult_\alpha} \cdot \prod\nolimits_{\zeta\ne\alpha} \abs{(m-r)-\zeta}^{\mult_\zeta} \notag\\
    &< 2^{-N} \cdot (\tfrac{1}{2}\sigma_\alpha)^{\mult_\alpha} \cdot \abs{\lcf R} \cdot \prod\nolimits_{\zeta\ne\alpha} (2\abs{\alpha-\zeta})^{\mult_\zeta} \notag\\
    &< \sigma_\alpha^{\mult_\alpha} \cdot \abs{\lcf R} \cdot \prod\nolimits_{\zeta\ne\alpha} \abs{\alpha-\zeta}^{\mult_\zeta} \notag\\
    &= \sigma_\alpha^{\mult_\alpha} \cdot \frac{\abs{R^{(\mult_\alpha)}(\alpha)}}{\mult_\alpha!} \notag\\
    &= \sigma_\alpha^{\mult_\alpha} \cdot 2^{\Oh(N+T)} \cdot M(\alpha)^N \notag\\
    \label{eq:LB-upper-bound}
    &= 2^{\Oh(\mult_\alpha(N+T))} M(\alpha)^N,
  \end{align}
  since \(R^{(\mult_\alpha)} / \mult_\alpha! \in \ZZ[x]\) has magnitude \((N, \Oh(N+T))\) and \(\sigma_\alpha \le \max_{\zeta\,:\,R(\zeta)=0} \abs{\zeta} = 2^{\Oh(N+T)}\) according to Cauchy's root bound,
  where the maximum is taken over all roots of \(R\) including \(\alpha\).

  We can also compute a lower bound for \(\LB(\alpha)\):
  \begin{align}
    \label{eq:LB-lower-bound}
    \LB(\alpha) &\ge \tfrac{1}{2} \cdot 2^{-\mult_\alpha - \deg R} \cdot \abs{R(m-r)} \notag\\
    &\ge 2^{-2N-1} \cdot \abs{\lcf R} \cdot (\tfrac{1}{32}\sigma_\alpha)^{\mult_\alpha} \cdot \prod\nolimits_{\zeta \ne \alpha} (\tfrac{1}{2}\abs{\alpha-\zeta})^{\mult_\zeta} \notag\\
    &> 2^{-8N} \cdot \sigma_\alpha^{\mult_\alpha} \cdot \abs{\lcf R} \cdot \prod\nolimits_{\zeta \ne \alpha} \abs{\alpha-\zeta}^{\mult_\zeta}.
  \end{align}
  Since we are mainly interested in a bound for the product over all \(\LB(\alpha)\), we first consider the product
  \begin{align*}
    \smash{\Pi \coloneqq \prod\nolimits_{\alpha} \big( 2^{-8N} \cdot \sigma_\alpha^{\mult_\alpha} \cdot \abs{\lcf R} \cdot \prod\nolimits_{\zeta \ne \alpha} \abs{\alpha-\zeta}^{\mult_\zeta} \big)}
  \end{align*}
  of the bound in \cref{eq:LB-lower-bound} over all \(\alpha\).
  Since \(\sum_\alpha \mult_\alpha \le N\), it follows that \(\prod_\alpha 2^{-8N} = 2^{\Oh(N^2)}\).
  For the product of the remaining factors, we first write the square-free decomposition of \(R\) as \(R = \prod_s r_s^s\) with square-free, pairwise coprime \(r_s \in \ZZ[x]\).
  Since \(R^{(s)}/s!\) has integer coefficients, we have
  \begin{align*}
    \smash{1 \le \abs[\Big]{\res\Big(r_s,\frac{R^{(s)}}{s!}\Big)} = \abs{\lcf r_s}^{\deg R - s} \cdot \prod_{\hspace{-1ex}\zeta\,:\, r_s(\zeta)=0\hspace{-1ex}} \frac{R^{(s)}(\zeta)}{s!}}
  \end{align*}
  and, thus,
  \begin{align*}
    &\hphantom{{}>} \prod\nolimits_{\alpha} \big(\sigma_\alpha^{\mult_\alpha} \cdot \abs{\lcf R} \cdot \prod\nolimits_{\zeta \ne \alpha} \abs{\alpha-\zeta}^{\mult_\zeta} \big)\\
    &> 2^{-\Sigma(R)} \cdot \prod\nolimits_\alpha \abs{\lcf R}\prod\nolimits_{\zeta\ne\alpha} \abs{\alpha-\zeta}^{\mult_\zeta}\\
    &= 2^{-\Sigma(R)} \cdot \prod\nolimits_\alpha \frac{\abs{R^{(\mult_\alpha)}(\alpha)}}{\mult_\alpha!}\\
    &= 2^{-\Sigma(R)} \cdot \prod\nolimits_s \abs{\lcf r_s}^{s-N} \cdot \abs[\Big]{\res\Big(r_s, \frac{R^{(s)}}{s!}\Big)}\\
    &\ge 2^{-\Sigma(R)} \cdot \abs{\lcf R} \cdot \abs{\lcf R^\ast}^{-N}\\
    &= 2^{-\sOh(N^2 + NT)},
  \end{align*}
  where we used that \(\Sigma(R) = \sOh (N^2+NT)\).
  Hence, \(\Pi\) is lower bounded by \(2^{-\sOh(N^2+NT)}\).
  Similar to the computation in \cref{eq:LB-upper-bound}, we can also determine an upper bound for the factor in \(\Pi\) corresponding to an arbitrary but fixed \(\alpha\).
  Namely, \(\sigma_\alpha^{\mult_\alpha} = 2^{\Oh(\mult_\alpha (N+T))}\) and
  \begin{align*}
    \smash{\abs{\lcf R} \prod_{\zeta \ne \alpha} \abs{\alpha-\zeta}^{\mult_\zeta} = \frac{\abs{R^{(\mult_\alpha)}(\alpha)}}{\mult_\alpha!} = 2^{\Oh(T+\log N)} M(\alpha)^N.}
  \end{align*}
  Thus, for any subset \(A\) of distinct roots of \(R\), the partial product
  \begin{align*}
    \Pi' \coloneqq \smash{\prod_{\alpha\in A}} \big( 2^{-8N} \cdot \sigma_\alpha^{\mult_\alpha} \cdot \abs{\lcf R} \cdot \prod\nolimits_{\zeta \ne \alpha} \abs{\alpha-\zeta}^{\mult_\zeta} \big)
  \end{align*}
  is upper bounded by \(2^{\sOh(N^2+NT)} \prod_{\alpha\in A} M(\alpha)^N = 2^{\sOh(N^2+NT)}\) since \(\prod_{\alpha\in A} M(\alpha) = 2^{\Oh(T)}\); cf.~\cref{bound-Mea}.
  We conclude that
  \begin{align*}
    \sum\nolimits_{\alpha}\LOG(\LB(\alpha)^{-1})\le \sOh(N^2+NT)+\sum\nolimits_{\alpha}\log(\LB(\alpha)^{-1}). \tag*{\qedhere}
  \end{align*}
\end{proof}

\clearpage
\section{Fast Approximate Polynomial Multipoint Evaluation}\label{appendix2}

Given a non-negative integer \(L\in\NN\), a polynomial \(F(x) = \sum_{i=0}^n F_i x^i \in \CC[x]\) of degree \(n\) and complex points \(x_1, \dots, x_n \in \CC\),
the task of approximate polynomial multipoint evaluation is to compute approximations \(\tilde{y}_j\) for \(y_j \coloneqq F(x_j)\) such that \(\abs{\tilde{y}_j - y_j} \le 2^{-L}\) for all \(j=1,\ldots,n\).
Let \(2^{\tau}\) and \(2^{\Gamma}\), with \(\tau,\Gamma\in\NN_{\ge 1}\), denote bounds on the absolute values of the coefficients of \(F\) and the points \(x_{j}\), respectively.

We aim to show that, using approximate arithmetic in the classical fast polynomial multipoint evaluation algorithm from~\cite{moenckborodin72} (see also~\cite[Sec.~10.1]{vzGG13}), we can compute approximations \(\tilde{y}_i\) as above with \(\sOh(n(L+\tau+n\Gamma))\) bit operations, and it suffices to consider \(L + \sOh(\tau + n\Gamma)\) bits of the coefficients of \(F\) and the points \(x_i\).
In particular, if \(L\) dominates \(\tau\) and \(n\Gamma\), the precision demand is essentially linear in \(L\), and the computation time is linear in \(L\) and \(n\).

This fact has been observed previously by Kirrinnis \cite[Thm.~3.9 and App.~A.3]{DBLP:journals/jc/Kirrinnis98} in a slightly different context.
Unfortunately, the formulation therein is less general and requires some transformations of the input, and the result does not seem to be widely known in the community.
Both Kirrinnis' and our discussion in \cite{KS13-multipoint} rely on a fast numerical polynomial division scheme using fast Fourier transforms due to Schönhage \cite{Schoenhage82}.
Existing implementations suggest that this algorithm is only efficient for extraordinarily large inputs.
Thus, we describe a different analysis based on polynomial division via Newton's method for polynomial inversion, which is known to have a more moderate break-even point.
To the best of our knowledge, this scheme has not been extensively analyzed in a numerical setting before, and we believe this section to be of independent interest.

\newparagraph
For the sake of simplicity, assume that \(n = 2^k\) is a power of two; otherwise, pad \(F\) with zeros.
We require that arbitrarily good approximations of the coefficients \(F_{i}\) and the points \(x_j\) are provided by an oracle for the cost of reading the approximations.
That is, asking for an approximation of the coefficients of \(F\) and all points \(x_{j}\) to a precision of \(\ell\) bits after the binary point takes \(\Oh(n(\tau+\ell))\) and \(\Oh(n(\Gamma+\ell))\) bit operations, respectively.

\begin{algorithm}[Multipoint evaluation]
  We will follow the classical di\-vide-and-con\-quer method for fast polynomial multipoint evaluation~\cite{moenckborodin72,vzGG13}:
  \label{alg:exact-mp-eval}%
  \begin{enumerate}%
  \item\label{alg:exact-mp-eval:subproducts}%
    From the linear factors \(g_{0,j}(x) \coloneqq x - x_j\), we recursively compute the subproduct tree
    \begin{align}\label{def:gij}
      g_{i,j}(x) &\coloneqq (x - x_{(j-1) 2^i + 1}) \cdots (x - x_{j 2^i})
      = g_{i-1,2j-1}(x) \cdot g_{i-1,2j}(x)
    \end{align}
    for \(i\) from 1 to \(k-1\) and \(j\) from \(1\) to \(n / 2^i=2^{k-i}\), that is, going up from the leaves.\linebreak[1]
    Notice that \(\deg g_{i,j} = 2^i\).
  \item\label{alg:exact-mp-eval:remainders}%
    Starting with \(r_{k,1}(x) \coloneqq F(x)\), we recursively compute the remainder tree
    \begin{align*}
      r_{i,j}(x) &\coloneqq F(x) \bmod g_{i,j}(x)
      = r_{i+1,\ceil{j/2}}(x) \bmod g_{i,j}(x)
    \end{align*}
    for \(i\) from \(k-1\) to \(0\) and \(j\) from \(1\) to \(n / 2^i=2^{k-i}\), that is, going down from the root.\linebreak[1]
    Notice that \(\deg r_{i,j} < 2^i\).
  \item\label{alg:exact-mp-eval:results}%
    Observe that the value at point \(x_j\) is exactly the remainder
    \begin{align*}
      r_{0,j} = F(x) \bmod g_{0,j}(x) = F(x) \bmod (x - x_j) = F(x_j) \in \CC.
    \end{align*}
  \end{enumerate}
\end{algorithm}

For the polynomial division with remainder, we use an asymptotically fast recursive approach, often called Newton's method for polynomial inversion.
It relies on Hensel lifting to compute the inverse of the \emph{reverse} polynomial (see page \pageref{def:reverse-polynomial} for the definition) of the divisor modulo some power of \(x\), which translates to the quotient in the original division.

\begin{algorithm}
  \label{alg:exact-div-rem}
  Given a polynomial \(F = \sum_{i=0}^{2n} F_i x^i \in \CC[x]\) of degree at most \(2n\) and a \emph{monic} polynomial \(G = \sum_{i=0}^n G_ix^i \in \CC[x]\) of degree \(n\),
  we compute the quotient \(Q\) and the remainder \(R\) of the polynomial division of \(F\) by \(G\) in the following way:
  \begin{enumerate}
  \item Define \(f \coloneqq \rev_{2n}F\),\quad \(g \coloneqq \rev_n G\),\quad \(h_0 \coloneqq 1\),\quad and \(k \coloneqq \ceil{\log (n+1)}\).
  \item For \(i=1, \dots, k\), recursively compute \(h_i \coloneqq 2\,h_{i-1} - g\cdot h_{i-1}^2 \bmod \cramped{x^{2^i}}\).
  \item Compute \(q \coloneqq f\cdot h_k \bmod x^{n+1}\),\quad \(Q = \rev_n q\),\quad and \(R = F - Q\cdot G\).
  \end{enumerate}
\end{algorithm}

The correctness of the above algorithm follows from the loop invariant that \(h_i\) is a multiplicative inverse of \(g\) modulo \(\cramped{x^{2^i}}\):
Observe that \(h_0\cdot g \equiv 1 \mod x\) since \(g\) has constant coefficient \(1\).
By definition, \(h_i\cdot g \equiv 1 - (h_{i-1}\cdot g-1)^2 \mod \cramped{x^{2^i}}\), and thus, by induction, we have \(h_{i-1}\cdot g-1 \equiv 0 \mod \cramped{x^{2^{i-1}}}\). It follows that \(h_i\cdot g \equiv 1 \mod \cramped{x^{2^i}}\).
It is now straightforward to verify that \(R\) as defined in Step 3 is of degree at most \(n-1\), and thus, \(Q\) and \(R\) are indeed the unique quotient and remainder of the division of \(F\) by \(G\).

We further remark that \(h_i \equiv h_{i-1} \mod \cramped{x^{2^{i-1}}}\) for all \(i > 0\).

\newparagraph
The arithmetic complexity of \cref{alg:exact-div-rem}, counting exact additions and multiplications in \(\CC\), is \(\Oh(\mul(n))\),
where \(\mul(n)\) denotes the arithmetic complexity of multiplication of two \(n\)-th degree polynomials.
For \cref{alg:exact-mp-eval}, it follows that, in the \(i\)-th layer of the subproduct tree and the remainder tree,
a total number of \(2^{k-i} \cdot \Oh(\mul(2^i)) = \Oh(\mul(n))\) field operations suffices.
Thus, the arithmetic complexity of multipoint evaluation is \(\Oh(\mul(n)\cdot\log n)\), which simplifies to \(\sOh(n)\)
if quasi-linear time polynomial multiplication algorithms are used~\cite{SS71,Furer09,DKSS08}.

To derive bounds on the bit complexity of these methods, when applied with approximate arithmetic, we recall the bit complexity and precision demand of approximate polynomial multiplication.
In the following considerations, we stipulate that \(\tilde{P}\) is an \(L\)-(bit) approximation of some \(P\) if \(\smash{\cramped{\norm{\tilde{P}-P}_1 \le 2^{-L}}}\).

\begin{lemma}
  \label{appr-mul-compl}
  Let \(F\) and \(G \in \CC[x]\) be polynomials of magnitude bounded by \((n, \tau)\), where \(\tau \in \NN\).
  Computing an \emph{\(\ell\)-bit approximation of} \(H = F\cdot G\), that is, computing an \(\tilde{H} \in \CC[x]\) such that \(\norm{\tilde{H}-H}_1 \le 2^{-\ell}\), is possible in
  \begin{align*}
    \Oh (\mul (n (\ell + \tau + 2\log n))) \quad&\text{or}\quad \sOh (n(\ell + \tau))
  \end{align*}
  bit operations and with a precision demand of at most
  \(
    \ell + \Oh (\tau + \log n)
  \)
  bits on each of the coefficients of \(F\) and \(G\).
\end{lemma}
\begin{proof}
  Let \(s \coloneqq \ell + \tau+2\ceil{\log(n+1)}+2\).
  Define \(f \coloneqq  2^s F\) and \(g \coloneqq  2^s G\), and notice that \(h \coloneqq f\cdot g = 2^{2s}H\).
  We consider polynomials \(\tf\) and \(\tg \in \ZZ[\ii][x]\) obtained from \(f\) and \(g\) by truncating the coefficients after the binary point,
  and write \(\Delta f \coloneqq \tf - f\) and \(\Delta g \coloneqq \tg - g\).
  Since \(\sumnorm{\Delta f}\), \(\sumnorm{\Delta g}\le n+1\) by definition of \(\tf\) and \(\tg\),
  \begin{alignat*}{2}
    \sumnorm{ \tf \, \tg - f\,g }
    &\le \sumnorm{\Delta f}\cdot \sumnorm{g}+\sumnorm{f}\cdot \sumnorm{\Delta g}+\sumnorm{\Delta f}\cdot \sumnorm{\Delta g}\\
    &\le (n+1)^{2}\cdot 2^{s+\tau}+(n+1)^{2}\cdot 2^{s+\tau}+(n+1)^{2}
    &\le (n+1)^{2}\cdot 2^{s+\tau+2}
  \end{alignat*}
  holds.
  For \(\tH \coloneqq 2^{-2s} \tf\,\tg\), it follows that
  \begin{align*}
    \sumnorm{\tH - H} \le 2^{-2s} (n+1)^{2} \cdot 2^{s+\tau+2}  \le 2^{\tau+2\log(n+1)+2 - s} \le 2^{-\ell},
  \end{align*}
  hence an \(\ell\)-bit-approximation as required can be recovered from the exact product of \(\tf\) and \(\tg\) by mere bitshifts.
  Since \(\infnorm{\tf}\), \(\infnorm{\tg} \le 2^{s+\tau}\), multiplication of \(\tf\) and \(\tg\) can be carried out exactly in \(\Oh (\mul((s+\tau) n))\) bit operations. This proves the complexity result.
  For the precision requirement, notice that \(\infnorm{f}\), \(\infnorm{g} \le 2^{s+\tau}\);
  thus, we need \((s + \tau + \ceil{\log (n+1)} + 3) \le (\ell + 2\tau + 3\LOG n + 8)\)-bit-approximations of the coefficients of \(F\) and \(G\) to compute \(\tf\) and \(\tg\).
\end{proof}

Notice that the norm of the polynomial factors affects the absolute precision of their product.
Hence, in order to evaluate the accuracy of the remainders in \cref{alg:exact-mp-eval}, we need good estimates on the norm of \(Q\) and \(R\) in \cref{alg:exact-div-rem}.
A naive bound for the coefficient growth in step 2 of \cref{alg:exact-div-rem} turns out to be too pessimistic.
Instead, we give a slightly generalized version of a result due to Sch\"{o}nhage \cite[Thm.~4.1]{Schoenhage82}:
\begin{lemma}
  \label{norm-Q-R}
  Let \(F = \sum_{i=0}^{2n} F_i x^i \in \CC[x]\) be a polynomial of degree at most \(2n\) and \(G = \sum_{i=0}^n G_ix^i \in \CC[x]\) be a \emph{monic} polynomial of degree \(n\).
  Let \(\rho \ge 1\) be an upper bound on the magnitude of the roots of \(G\).
  If \(Q\) and \(R\) are the quotient and the remainder of the division of \(F\) by \(G\), that is \(F = Q\cdot G + R\) with \(R\) uniquely defined by \(\deg R < \deg G = n\), then it holds that
  \begin{align*}
    \norm{Q}_1 \le 2^{2n} \rho^{n} \cdot \norm{F}_1
    \qquad\text{and}\qquad
    \norm{R}_1 \le 2^{4n} \rho^{4n} \cdot \norm{F}_1.
  \end{align*}
\end{lemma}
\begin{proof}
  The coefficients \(Q_k\) of \(Q\) appear as the leading coefficients in the Laurent series of the function
  \begin{align*}
    \frac{F(x)/x^n}{G(x)} = \frac{F_{2n} + F_{2n-1}/x + F_{2n-2}/x^2 + \cdots}{G_n + G_{n-1}/x + G_{n-2}/x^2 + \cdots}
    = Q_n + \frac{Q_{n-1}}{x} + \frac{Q_{n-2}}{x^2} + \cdots ,
  \end{align*}
  and thus, using Cauchy's integral formula, they can be represented as
  \begin{align}
    \label{eq:cauchy-integral}
    Q_k = \frac{1}{2\pi \ii} \int_{\abs{x}=\varrho} \frac{F(x)/x^n}{G(x)} \, x^{n-k-1} \; dx
  \end{align}
  with an arbitrary positive number \(\varrho > \rho\); see~\cite[(4.7)--(4.9)]{Schoenhage82}.
  For an arbitrary \(x\) on the boundary of the disk \(D_\varrho(0)\) with radius \(\varrho\) centered at the origin, it holds that \(\abs{F(x) \, x^{-k-1}} \le \norm{F}_1 \cdot (2\varrho)^{2n-k-1}\le \norm{F}_1\cdot (2\varrho)^{2n-1}\) and \(\abs{G(x)} \ge (\varrho - \rho)^n\) because \(G\) is monic and the distance from \(x\) to any root of \(F\) is at least \(\varrho-\rho\). Hence, substitution of \(\varrho = 2\rho\) in \cref{eq:cauchy-integral} yields
  \begin{align*}
    \abs{Q_k} \le 2^n \cdot (2\rho)^{n} \cdot \norm{F}_1.
  \end{align*}
  This proves the bound on \(\norm{Q}_1\).
  The second claim now follows from \(\norm{R}_1 \le \norm{F}_1 + \norm{Q}_1 \cdot\norm{G}_1\), using the triangle inequality and the submultiplicativity of \(\norm{\cdot}_1\),
  and the observation that the magnitude of the \(i\)-th coefficient of \(G\) is bounded by \(\binom{n}{i}\,\rho^i\) and, hence, \(\smash{\norm{G}_1 \le \sum \binom{n}{i}\,\rho^i = (1+\rho)^n \le (2\rho)^n}\).
\end{proof}

We are now in the position to assemble the main results for the bit complexity of approximate division with remainder for monic divisors and, in turn, of multipoint evaluation.

\begin{corollary}
  \label{div-compl-monic}
  Let \(F \in \CC[x]\) be a polynomial of magnitude bounded by \((2n, \tau)\) and \(G \in \CC[x]\) be a \emph{monic} polynomial of degree \(n\) with a given upper bound \(2^\Gamma \ge 1\) on the magnitude of the roots of \(G\), where \(\Gamma \ge 1\).
  Computing an \(\ell\)-bit approximation of \(R = F \bmod G\), that is, computing an \(\tilde{R} \in \CC[x]\) such that \(\norm{\tilde{R}-R}_1 \le 2^{-\ell}\), is possible in
  \begin{align*}
    \sOh (n(\ell+\tau + n\Gamma))
  \end{align*}
  bit operations and with a precision demand of at most
  \(
    \ell + \sOh (\tau + n\Gamma)
  \)
  bits on each of the coefficients of \(F\) and \(G\).
\end{corollary}
\begin{proof}
  We apply \cref{alg:exact-div-rem} in approximate arithmetic with operands \(\tilde{\cdot}\) and, if necessary, discard the leading terms of the result \(\tilde{R}\) of degree higher than \(n-1\).
  The analysis of the required precision is done in a backwards fashion.
  Recall that \(\norm{G}_1 \le (1+2^\Gamma)^n \le 2^{\Oh(n\Gamma)}\).

  \Cref{appr-mul-compl}, applied to the computation of \(R\) in Step~3, shows that the computation of \(\tilde{R}\) entails an error of at most \(2^{-\ell}\)
  if \(\tilde{F}\), \(\tilde{G}\) and \(\tilde{Q}\) are \(L\)-approximations to the exact polynomials for some \(L = \ell + \sOh(\tau + n\Gamma)\).
  Accordingly, in Step~2, we need to obtain an \(L'\)-approximation of \(h_k\) with some \(L'= L + \sOh(\tau + \LOG\,\norm{h_k}_1)\).

  Observe that \(\rev_n h_k\) is exactly the quotient of the division of \(x^{2n}\) by \(G\), and thus, \(\LOG \norm{h_k}_1 = \Oh(n\Gamma)\) by \Cref{norm-Q-R}. It follows that \(\LOG\,\norm{h_i}_1 = \Oh(n\Gamma)\) for all \(i=0,\ldots,k\) as the intermediate inverses \(h_i\) are identical to \(h_k \bmod \cramped{x^{2^i}}\).
  We conclude that each of the \(k\) iterations during Step~2 deteriorates the accuracy by at most \(\Oh(n\Gamma+\log n)\) bits.
  Hence, it follows that it suffices to consider \(L''\)-bit approximations of \(F\) and \(G\), with \(L'' = L' + \ceil{\log(n+1)}\cdot \Oh(n\Gamma+\log n) = \ell + \sOh(\tau + n\Gamma)\), in order to eventually obtain an \(\ell\)-bit approximation for \(\tilde{R}\).

  For the bit complexity bound, notice that it suffices to run \cref{alg:exact-div-rem} with fixed precision arithmetic and an accuracy of \(L''\) bits after the binary point,
  where each of the \(\sOh(n)\) field operations in \(\CC\) requires \(\sOh (L'') = \sOh(\ell + \tau + n\Gamma)\) bit operations.
\end{proof}

\begin{theorem}
  \label{apx-mp-eval-complexity}
  Let \(F \in \CC[x]\) be a polynomial of magnitude \((n,\tau)\), with \(\tau \ge 1\), and let \(x_1, \dots, x_n \in \CC\) be complex points with absolute values bounded by \(2^{\Gamma}\), where \(\Gamma\ge 1\).
  Then, for an arbitrary positive integer \(L\), the computation of values \(\tilde{y}_j\) with \(\abs{\tilde{y}_j - F(x_j)} \le 2^{-L}\) for all \(j\), needs
  \begin{align*}
    \sOh (n (L + \tau + n\Gamma))
  \end{align*}
  bit operations.
  Moreover, the precision demand on \(F\) and the points \(x_j\) is bounded by \(L + \sOh (\tau + n\Gamma)\) bits.
\end{theorem}
\begin{proof}
  Define \(g_{i,j}\) and \(r_{i,j}\) as in \cref{alg:exact-mp-eval}.
  We analyse a run of the algorithm using approximate multiplication and division, with a precision of \(\ell^{\bdiv}_i\) for the approximate divisors \(\tg_{i,\ast}\) and remainders \(\tr_{i,\ast}\)
  in the \(i\)-th layer of the subproduct and the remainder tree.
  We recall that \(\deg \tg_{i,\ast} = \deg g_{i,\ast} = 2^i\).

  According to \cref{div-compl-monic}, for the recursive divisions to yield an output precision \(\ell_i \ge 0\),
  it suffices to have approximations \(\tr_{i+1, \ast}\) and \(\tg_{i,\ast}\) of the exact polynomials \(f\coloneqq r_{i+1,\ast}\) and \(g\coloneqq g_{i,\ast}\) to a precision of
  \begin{align}
    \label{eq:prec-l_i}
    \ell^{\bdiv}_{i+1} &\coloneqq \ell^{\bdiv}_i + \sOh( \log \sumnorm{r_{i+1,\ast}} +2^i\Gamma )
  \end{align}
  bits, since the roots of each \(g_{i,\ast}\) are contained in the set \(\{x_{1},\ldots,x_{n}\}\) and, thus, their absolute values are also bounded by \(2^{\Gamma}\).
  A bound the magnitude of the remainders \(r_{i,\ast}\) for \(i<\log n\) is a consequence of \cref{norm-Q-R}, applied in an iterative manner, which yields
  \begin{align}
    \label{eq:bound-r_i}
    \log \sumnorm{r_{i,\ast}} &= \log \sumnorm{r_{i+1,\ast}} + \Oh(2^{i+1}\Gamma + i\cdot 2^{i})
    = \Oh(\tau +2n\Gamma + n\log n).
  \end{align}
  For the last estimation, we use that \(\sumnorm{r_{\log n,0}}=\sumnorm{F}\le (n+1)2^{\tau}\).

  Combining \cref{eq:prec-l_i} and \cref{eq:bound-r_i} yields
  \begin{inlineableEquation}
    \smash{\cramped{\ell^{\bdiv} \coloneqq \max_{i>0} \ell^{\bdiv}_i = \ell^{\bdiv}_0 + \Oh(\tau +2n\Gamma+ n\log n)}}.
  \end{inlineableEquation}
  Hence, choosing \(\smash{\cramped{\ell^{\bdiv}_0\coloneqq L,}}\) we eventually achieve evaluation up to an error of \(2^{-L}\) if all numerical divisions are carried out with precision \(\ell^{\bdiv}\).
  The bit complexity to carry out a single numerical division at the \(i\)-th layer of the tree is then bounded by \(\sOh(2^{i}(\ell^{\bdiv}+\tau+2^{i}\Gamma))=\sOh(2^{i}(L+n\Gamma+\tau))\).
  Since there are \(n/2^{i}\) divisions, the total cost at the \(i\)-th layer is bounded by \(\sOh(n(L+n\Gamma+\tau))\).
  The depth of the tree equals \(\log n\), and thus the overall bit complexity is \(\sOh(n(L+n\Gamma+\tau))\).

  It remains to bound the precision demand and, hence, the cost for computing \((L+\tau+2n\Gamma+\Oh(n\log n))\)-bit approximations of the polynomials \(g_{i,\ast}\).
  According to \cref{appr-mul-compl}, in order to compute the polynomials \(g_{i,\ast}\) to a precision of \(\ell^{\operatorname{mul}}_i\), we have to consider \(\ell^{\operatorname{mul}}_{i-1}\)-bit approximations of \(g_{i-1,\ast}\),  where
  \begin{align*}
    \ell^{\operatorname{mul}}_i = \ell^{\operatorname{mul}}_{i-1} + \Oh (\LOG \sumnorm{g_{i-1,\ast}} + i)=\ell^{\operatorname{mul}}_{i-1} + \Oh (i\,\Gamma)=\ell^{\operatorname{mul}}_{0}+\sOh(\log n\cdot \Gamma).
  \end{align*}
  Hence, it suffices to run all multiplications in the product tree with a precision of \(\ell^{\operatorname{mul}}=L+\sOh(\tau+n\Gamma+n\log n)\).
  The bit complexity for all multiplications is bounded by \(\sOh(n\,\ell^{\operatorname{mul}})=\sOh(n(L+\tau+n\Gamma))\),
  and the precision demand for the points \(x_{i}\) is bounded by \(\ell^{\operatorname{mul}}+O(\Gamma+\log n)=L+\sOh(\tau+n\Gamma+n\log n)\).
\end{proof}

\clearpage
\renewcommand*{\UrlFont}{\ttfamily\relsize{-0.5}\relax}
\printbibliography
\end{document}